\documentclass[twoside]{IEEEtran}

\usepackage{amsmath,amssymb,eucal}
\usepackage{xifthen}
\usepackage{mathtools}
\usepackage{enumerate}
\usepackage{microtype}
\usepackage{xspace}
\usepackage{bm}
\usepackage[T1]{fontenc}
\usepackage{fancyhdr}
\usepackage{lastpage}
\usepackage{color}

\mathtoolsset{showonlyrefs}

\newcommand{\mbs}[1]{\bm{#1}}
\newcommand{\vect}[1]{{\lowercase{\mbs{#1}}}}
\newcommand{\mat}[1]{{\uppercase{\mbs{#1}}}}

\newcommand{\T}{{\scriptscriptstyle\mathsf{T}}}
\renewcommand{\H}{{\scriptscriptstyle\mathsf{H}}}

\newcommand{\cond}{\,\vert\,}
\renewcommand{\Re}[1][]{\ifthenelse{\isempty{#1}}{\operatorname{Re}}{\operatorname{Re}\left(#1\right)}}
\renewcommand{\Im}[1][]{\ifthenelse{\isempty{#1}}{\operatorname{Im}}{\operatorname{Im}\left(#1\right)}}

\newcommand{\bv}{\vect{b}}

\newcommand{\gv}{\vect{g}}
\newcommand{\hv}{\vect{h}}

\newcommand{\sv}{\vect{s}}

\newcommand{\uv}{\vect{u}}
\newcommand{\vv}{\vect{v}}

\newcommand{\xv}{\vect{x}}
\newcommand{\yv}{\vect{y}}
\newcommand{\zv}{\vect{z}}

\newcommand{\nuv}{\vect{\nu}}

\newcommand{\psiv}{\vect{\psi}}
\newcommand{\etav}{\vect{\eta}}

\newcommand{\Thetam}{\pmb{\Theta}}
\newcommand{\Omegam}{\pmb{\Omega}}

\newcommand{\Lambdam}{\pmb{\Lambda}}
\newcommand{\Phim}{\pmb{\Phi}}
\newcommand{\Psim}{\pmb{\Psi}}

\newcommand{\Am}{\mat{a}}
\newcommand{\Bm}{\mat{b}}
\newcommand{\Cm}{\mat{c}}
\newcommand{\Dm}{\mat{d}}

\newcommand{\Km}{\mat{k}}

\newcommand{\Qm}{\mat{q}}

\newcommand{\Sm}{\mat{s}}

\newcommand{\Vm}{\mat{V}}

\newcommand{\Cc}{{\mathcal C}}

\newcommand{\Sc}{{\mathcal S}}

\newcommand{\Xc}{{\mathcal X}}

\newcommand{\CC}{\mathbb{C}}

\newcommand{\Id}{\mat{\mathrm{I}}}

\newcommand{\CN}[1][]{\ifthenelse{\isempty{#1}}{\mathcal{N}_{\mathbb{C}}}{\mathcal{N}_{\mathbb{C}}\left(#1\right)}}

\renewcommand{\P}[1][]{\ifthenelse{\isempty{#1}}{\mathbb{P}}{\mathbb{P}\left(#1\right)}}
\newcommand{\E}[1][]{\ifthenelse{\isempty{#1}}{\mathbb{E}}{\mathbb{E}\left(#1\right)}}
\renewcommand{\det}[1][]{\ifthenelse{\isempty{#1}}{\text{det}}{\text{det}\left(#1\right)}}
\newcommand{\trace}[1][]{\ifthenelse{\isempty{#1}}{\text{tr}}{\text{tr}\left(#1\right)}}
\newcommand{\rank}[1][]{\ifthenelse{\isempty{#1}}{\text{rank}}{\text{rank}\left(#1\right)}}
\newcommand{\diag}[1][]{\ifthenelse{\isempty{#1}}{\text{diag}}{\text{diag}\left(#1\right)}}

\DeclarePairedDelimiter\abs{\lvert}{\rvert}
\DeclarePairedDelimiter\Abs{\lvert}{\rvert^2}
\DeclarePairedDelimiter\norm{\lVert}{\rVert}
\DeclarePairedDelimiter\Norm{\lVert}{\rVert^2}

\newcommand{\defeq}{\triangleq}

\newtheorem{proposition}{Proposition}
\newtheorem{remark}{Remark}[section]
\newtheorem{definition}{Definition}
\newtheorem{theorem}{Theorem}
\newtheorem{corollary}{Corollary}
\newtheorem{lemma}{Lemma}
\newtheorem{assumption}{Assumption}

\newcommand{\Rmimo}{R_{\text{mimo}}}
\newcommand{\dmimo}{d_{\text{mimo}}}
\newcommand{\Rmimoi}{R_{\text{mimo},1}}
\newcommand{\Rmimoj}{R_{\text{mimo},2}}
\newcommand{\Rmimok}{R_{\text{mimo},k}}

\newcommand{\Retav}{R_{\etav}}

\newcommand{\Qu}{\Qm_{\uv}}
\newcommand{\Qv}{\Qm_{\vv}}
\newcommand{\Psimg}{\Psim_{\hat{\gv}}}
\newcommand{\Psimgp}{\Psim_{\hat{\gv}^{\perp}}}
\newcommand{\Psimh}{\Psim_{\hat{\hv}}}
\newcommand{\Psimhp}{\Psim_{\hat{\hv}^{\perp}}}

\newcommand{\noisey}{\varepsilon}
\newcommand{\noisez}{\omega}
\newcommand{\Noisey}{E}
\newcommand{\Noisez}{\Omega}

\begin{document}

\title{Degrees of Freedom of Time Correlated MISO Broadcast Channel with Delayed CSIT}
\author{
\thanks{Manuscript submitted to IEEE Transactions on Information Theory
in March 2012, revised in August 2012.}
Sheng~Yang,~\IEEEmembership{Member,~IEEE,} Mari~Kobayashi,~\IEEEmembership{Member,~IEEE,} 
\thanks{S. Yang and M. Kobayashi are with the Telecommunications
department of SUPELEC, 3 rue Joliot-Curie, 91190 Gif-sur-Yvette,
France.~(e-mail: \texttt{\{sheng.yang, mari.kobayashi\}@supelec.fr})} \\
David~Gesbert,~\IEEEmembership{Fellow,~IEEE,} 
Xinping~Yi,~\IEEEmembership{Student Member,~IEEE} 
\thanks{D. Gesbert and X. Yi are with EURECOM,
Sophia-Antipolis, France.~(e-mail: \texttt{\{david.gesbert, xinping.yi\}@eurecom.fr})}
\thanks{This work was partially supported by HIATUS and the ANR project
FIREFLIES~(ANR-10-INTB-0302). The project HIATUS acknowledges the financial support of the Future and Emerging Technologies (FET) programme within the Seventh Framework Programme for Research of the European Commission under FET-Open grant number: 265578.}
\thanks{
Parts of the results have been presented in IEEE International Symposium
on Information Theory, Boston, USA, July 2012.}}

\maketitle

\begin{abstract}
We consider the time correlated multiple-input single-output~(MISO)
broadcast channel where the transmitter has imperfect knowledge of the
current channel state, in addition to delayed channel state information.
By representing the quality of the current channel state information as
$P^{-\alpha}$ for the signal-to-noise ratio $P$ and some constant
$\alpha\geq 0$, we characterize the optimal degree of freedom region for
this more general two-user MISO broadcast correlated channel.  The
essential ingredients of the proposed scheme lie in the quantization and
multicast of the overheard interferences, while broadcasting new private
messages. Our proposed scheme smoothly bridges between the scheme
recently proposed by Maddah-Ali and Tse with no current state
information and a simple zero-forcing beamforming with perfect current
state information.
\end{abstract}

\section{Introduction}

In most practical scenarios, perfect channel state information at
transmitter (CSIT) may not be available due to the time-varying nature
of wireless channels as well as the limited resource for channel
estimation. However, many wireless applications must guarantee high-data
rate and reliable communication in the presence of channel uncertainty.
In this paper, we consider such a scenario in the context of the two-user
multiple-input single-output~(MISO) broadcast channel, where the
transmitter equipped with $m$ antennas~($m\ge2$) wishes to send two private messages to two receivers each with a single
antenna.  The discrete time signal model is given by 
\begin{subequations}
  \begin{align}
    y_t &= \hv_t^\H \xv_t + \noisey_t, \label{eq:signal-model-a}\\
    z_t &= \gv_t^\H \xv_t + \noisez_t, \label{eq:signal-model-b}
  \end{align}
\end{subequations}
for any time instant $t$, where $\hv_t,\gv_t \in \CC^{m\times 1}$ are
the channel vectors for user~1 and user~2, respectively; $\noisey_t, \noisez_t \sim
\CN[0,1]$ are normalized additive white Gaussian noises~(AWGN) at the
respective receivers; the input signal $\xv_t$ is subject to the power
constraint $\E\bigl( \norm{\xv_t}^2 \bigr) \le P$, $\forall\,t$.  

For the case of perfect CSIT, the optimal degrees of freedom~(DoF) of this
channel is two and achieved by linear strategies such as
zero-forcing~(ZF) beamforming. When the transmitter suffers from
constant inaccuracy of channel estimation, it has been shown in
\cite{Wigger} that the degrees of freedom per user is
upper-bounded by $\frac{2}{3}$, whereas the highest known achievable DoF
value, also conjectured to be optimal, is only $\frac{1}{2}$. 
It is also well known that the full
multiplexing gain can be maintained under imperfect CSIT if the error in
CSIT decreases as $O(P^{-1})$ as $P$ grows
\cite{caire2010multiuser}.  Moreover, for the case of the temporally
correlated fading channel such that the transmitter can predict the
current state with error decaying as $O(P^{-\alpha})$ for some constant
$\alpha\in[0,1]$, ZF can only achieve a fraction $\alpha$ of the optimal
degrees of freedom \cite{caire2010multiuser}.  This result somehow
reveals the bottleneck of a family of precoding schemes relying only on
instantaneous CSIT as the temporal correlation decreases
($\alpha\rightarrow 0$). Recently, a breakthrough has been made in order
to overcome this problem. In \cite{maddah2010degrees}, Maddah-Ali and
Tse showed a surprising result that even completely outdated CSIT can be
very useful in terms of degrees of freedom, as long as it is accurate.
For a system with $m\ge2$ antennas and two users, the proposed scheme in
\cite{maddah2010degrees}, hereafter called MAT, achieves the
multiplexing gain of $\frac{2}{3}$ per user, irrespectively of the
temporal correlation. 
The role of perfect
delayed CSIT can be re-interpreted as a feedback of the past
signal/interference heard by the receivers. This side information
enables the transmitter to perform ``retrospective'' alignment in the
space \emph{and} time domain, as demonstrated in different multiuser
network systems (see \cite{FandT_Jafar} and the references therein).  Despite its DoF
optimality, the MAT scheme is designed assuming the worst case scenario
where the delayed channel feedback provides no information about the
current channel state. This assumption is over pessimistic as most practical
channels exhibit some form of temporal correlation.  In fact, it readily
follows that the selection strategy between ZF and MAT yields the degrees
of freedom of $\max\{\alpha, \frac{2}{3}\}$ for $\alpha\in[0,1]$.  For
either quasi-static fading channel~($\alpha \ge 1$) or very fast
channels~($\alpha \to 0$), a selection approach is reasonable. However,
for intermediate ranges of temporal correlation ($0<\alpha<1$), a
fundamental question arises as to whether a better way of exploiting
both delayed CSIT and current (imperfect) CSIT exists.  Studying the
DoF under such a CSIT assumption is of practical and
theoretical interest.

The main contributions of this work are summarized in the following.
First, we establish an outer bound on the DoF region of the two-user broadcast
channel with perfect delayed and imperfect current state information. To
that end, we use two powerful tools: the genie-aided model and the
extremal inequality~\cite{LiuViswanath,Weingarten}. 
Then, we propose a
novel scheme that optimally combines the ZF spatial precoding, based on the imperfect 
current state information, and the MAT space-time alignment, based on the perfect 
past state information. The key of this scheme is the digital
transmission of the overheard interference, which replaces the analog
one initially considered in the MAT alignment~\cite{maddah2010degrees}. The role of
spatial precoding, exploiting current CSIT, is two-fold: 
\begin{itemize}
  \item It enables to reduce the power of overheard interferences in the MAT alignment. This 
    power reduction then saves, via source compression/quantization,
    the resource related to the transmission of the overheard
    interferences.
  \item It allows for the parallel transmission of two private messages
    on top of the multicast of overheard interferences as common message. 
\end{itemize}
It will be shown that the proposed scheme achieves the upper bound of
the symmetric DoF  
\begin{align}
  d_{\text{sym}} = \frac{2+\alpha}{3}, \quad\alpha\in[0,1]
\end{align}%
given by the converse. To achieve the other corner points of the region,
we show that delayed CSIT is not necessary and the optimal strategy is
a combination of rate-splitting, spatial precoding with imperfect
current CSI, and superposition coding. 
Specifically, we split one of the users' message into two parts and
broadcast one part of it as common message. The other part and the
message of the other user are then superimposed over the common message
and broadcast with spatial precoding. 
As an extension to the main result, we derive the optimal DoF region of
the same channel with common message. Another extension is the
achievable DoF region when only imperfect delayed CSIT is
available~(e.g., due to limited feedback rates). 
Finally, in addition to the results on the optimal DoF region, we provide the exact achievable rate regions of the proposed schemes in the appendix. 

At the time of submission, a parallel independent work~\cite{Gou2012}
was brought to our attention which also builds on our initial results
reported in \cite{SubmittedISIT}. In \cite{Gou2012}, the authors consider an i.i.d.
fading model in which the transmitter knows perfectly the past channel
states and imperfectly the current channel state. Their achievability
proof coincides with our optimal scheme, while the outer bound is
derived differently by establishing an equivalent compound channel.  
It is worth noting that the outer bound technique developed in
\cite{Gou2012} does not rely on any essential statistical equivalence of
the two users' channel vector directions, which is stronger than both
the original result of \cite{maddah2010degrees} as well as the result in
 this work~(that exploits the isotropic property of the estimation
 error). On the other hand, our model allows temporal correlations of
 the channel coefficients and is therefore stronger than both the
 original result~\cite{maddah2010degrees} and \cite{Gou2012} in that
 sense. 
Thus, while both \cite{Gou2012} and the current work generalize \cite{maddah2010degrees}, neither subsumes the other.  

The rest of the paper is organized as follows. In
Section~\ref{sec:model}, after presenting the assumptions and some basic
definitions of our model, we provide our main theorem on the optimal DoF
region. The above contributions are then presented in order. Finally, we
conclude the paper in Section~\ref{sec:conclusions}. Detailed proofs are deferred
to the appendix. 
   
Throughout the paper, we will use the following notations. 
Matrix transpose, Hermitian
transpose, inverse, and determinant are denoted by $\Am^\T$,
$\Am^{\H}$, $\Am^{-1}$, and $\det[\Am]$, respectively. 
$\xv^{\perp}$ is any nonzero vector such that $\xv^\H \xv^\perp = 0$. 
Logarithm is in base $2$. Partial ordering of Hermitian matrices
is denoted by $\succeq$ and $\preceq$, i.e., $\Am\succeq\Bm$ means
$\Am-\Bm$ is positive semidefinite. We use $\Psim_{\xv}$ to denote a
projection matrix on the direction given by $\xv$, i.e., $\Psim_{\xv} \defeq
\displaystyle \frac{\xv \xv^\H}{\norm{\xv}^2}$.

\section{System Model and Main Results}
\label{sec:model}

The signal model of this paper is defined by \eqref{eq:signal-model-a}
and \eqref{eq:signal-model-b}. 
For convenience, we provide the following definition.
\begin{definition}[channel states]
  The channel vectors $\hv_t$ and $\gv_t$ are called the states of the
  channel at instant $t$. For simplicity, we also define the state
  matrix $\Sm_t$ as $\Sm_t \defeq \left[ \begin{smallmatrix}\hv_t^\H \\ \gv_t^\H
  \end{smallmatrix} \right] \in \mathcal{S}$ where $\mathcal{S}$ is the
  set of all possible states.
\end{definition}
The assumptions on the knowledge of the channel
states and the fading process are summarized as follows. 

\begin{assumption}[perfect delayed and imperfect current CSI]\label{assumption:CSI}
At each time instant $t$, the transmitter knows the delayed channel states up to
instant $t-1$. In addition, the transmitter can
somehow obtain an estimate $\hat{\Sm}_{t} \in \hat{\mathcal{S}}$ of the current channel state $\Sm_t$, i.e., $\hat{\hv}_t$
and $\hat{\gv}_t$ are available to the transmitter with
\begin{align}
  \hv_t &= \hat{\hv}_t + \tilde{\hv}_{t}, \\
  \gv_t &= \hat{\gv}_t + \tilde{\gv}_{t}  
\end{align}
where the estimate $\hat{\hv}_t$~(also $\hat{\gv}_t$) and 
estimation error $\tilde{\hv}_{t}$~(also $\tilde{\gv}_{t}$) are
uncorrelated and both assumed to be zero mean with covariance
$(1-\sigma^2) \Id_m$ and $\sigma^2 \Id_m$, respectively, with $\sigma^2
\le 1$. The receivers know perfectly all states $\bigl\{\Sm_t\bigr\}$ and
$\bigl\{\hat{\Sm}_t\bigr\}$. 
\end{assumption}


\begin{assumption}[fading process]
  \label{assumption:fading}
  The processes $\bigl\{ \hat{\Sm}_t \bigr\}$, $\bigl\{ \tilde{\Sm}_t
  \bigr\}$, and thus $\bigl\{ {\Sm}_t \bigr\}$ are stationary and
  ergodic. Moreover, for any time instant $t$, we assume the following: 
  \begin{enumerate}
    \item $\rank[\Sm_t] = 2$ with probability
      $1$ and $\E\bigl( \log\det[\Sm_t\Sm_t^\H] \bigr) > -\infty$. 
    \item We have the Markov chain  
      \begin{equation} 
        (\hat{\Sm}^{t-1}, {\Sm}^{t-1}) \leftrightarrow \hat{\Sm}_{t} \leftrightarrow {\Sm}_{t}.  \label{eq:Markov} 
      \end{equation}
    \item The estimation error is isotropic, i.e., the distributions of
      $\tilde{\hv}_t$ and $\tilde{\gv}_t$ conditional on $\hat{\Sm}_t$
      are invariant under unitary transformations. Furthermore, for any
      $\sigma^2>0$, 
      $\E_{\tilde{S}_{i} \vert \hat{S}_i} \Bigl( \log \frac{\Abs{\tilde{h}_{t,i}}}{\sigma^2} \Bigr)$ and $\E_{\tilde{S}_{i} \vert \hat{S}_i} \Bigl( \log \frac{\Abs{\tilde{g}_{t,i}}}{\sigma^2} \Bigr)$, $i=1,\ldots,m$, are finite. 
  \end{enumerate}
\end{assumption}
Note that when $\bigl\{ \hat{\Sm}_t \bigr\}$ and $\bigl\{ \tilde{\Sm}_t
\bigr\}$ are independent Rayleigh fading processes with independent and
identically distributed~(i.i.d.) entries, all the above
assumptions are verified.  Without loss of generality, we implicitly
assume that $\sigma^2>0$ in the rest of the paper.
The case with $\sigma^2 = 0$ corresponds to the case of perfect CSIT, in
which the capacity region is already known. Then, we can introduce
a parameter $\alpha_P\ge0$ as the power exponent of the estimation error
\begin{align}
  \alpha_P \defeq -\frac{\log(\sigma^2)}{\log P}. \label{eq:alpha}
\end{align}
The parameter $\alpha_P$ can be regarded as the quality of the
current CSIT in the high SNR regime. Note that $\alpha_P=0$ corresponds
to the case with no current CSIT at all, while $\alpha_P\to\infty$ corresponds to the case with perfect current
CSIT. In addition, we assume that $\displaystyle
\lim_{P\to\infty} \alpha_P$ exists and define
\begin{align}
\alpha \defeq \lim_{P\to\infty} \alpha_P.  
\end{align}%
Hereafter, we use $\alpha$ instead of $\alpha_P$, whenever no confusion
is likely. In addition, since $\alpha>1$ implies that the
estimation noise is negligible as compared to the AWGN and can be
regarded as perfect from the DoF perspective, we assume implicitly that the value of
$\alpha>1$ is truncated at $1$ wherever applicable.
Connections between the above model and practical time correlated models
are highlighted in Section~\ref{sec:extension}. 

\begin{definition}[achievable degrees of freedom]
  A code for the two-user Gaussian MISO broadcast channel with delayed CSIT and
  imperfect current CSIT is defined as follows: 
  \begin{itemize}
\item A sequence of encoders at time $t$ is given by 
  $F_t: \mathcal{W}_1 \times \mathcal{W}_2 \times \mathcal{S}^{t-1}
  \times \hat{\mathcal{S}}^{t} \longmapsto \CC^{m}$ where the messages
  $W_1$ and $W_2$ are uniformly distributed over the message sets $\mathcal{W}_1$
  and $\mathcal{W}_2$, respectively.
\item A decoder for user~$k$ is given by the mapping $\hat{W}_k:
  \mathbb{C}^{1 \times n} \times \mathcal{S}^{n} \times
  \hat{\mathcal{S}}^{n}  \longmapsto \mathcal{W}_k$, $k=1,2$.
\end{itemize}
The DoF pair $(d_1, d_2)$ is said to be \emph{achievable} if there exists a code
that simultaneously satisfies the reliability condition 
\begin{align}
\limsup\limits_{n\rightarrow \infty} & \Pr\bigl\{ W_k\neq \hat{W}_k
\bigr\}=
0,\label{eq:reliability} \\
\shortintertext{and has a pre-log factor of the rate}
\lim_{P\rightarrow \infty} \liminf_{n\rightarrow \infty} & \frac{
\log_2 |\mathcal{W}_k(n, P)|}{ n \log_2 P}\geq d_k, \quad k=1,2.  \label{eq:DoF} 
\end{align}
The union of all achievable DoF pairs is then called the optimal DoF 
region of the Gaussian MISO broadcast channel. 
\end{definition}

The main result of this paper is stated below. 
\begin{theorem} \label{theorem:DoF}
  The optimal degrees of freedom region of the two-user Gaussian MISO
  broadcast channel with perfect delayed and imperfect current CSIT is
  characterized by 
  \begin{subequations}
  \begin{align}
    d_1 &\le 1, \label{eq:d1}\\
    d_2 &\le 1, \label{eq:d2}\\
    d_1 + 2 d_2 &\le 2 + \alpha, \label{eq:d1+2d2}\\
    2 d_1 + d_2 &\le 2 + \alpha. \label{eq:2d1+d2} 
  \end{align}%
\label{eq:region}
  \end{subequations}
\end{theorem}
\begin{figure}
 \centering
\includegraphics[width=0.48\textwidth]{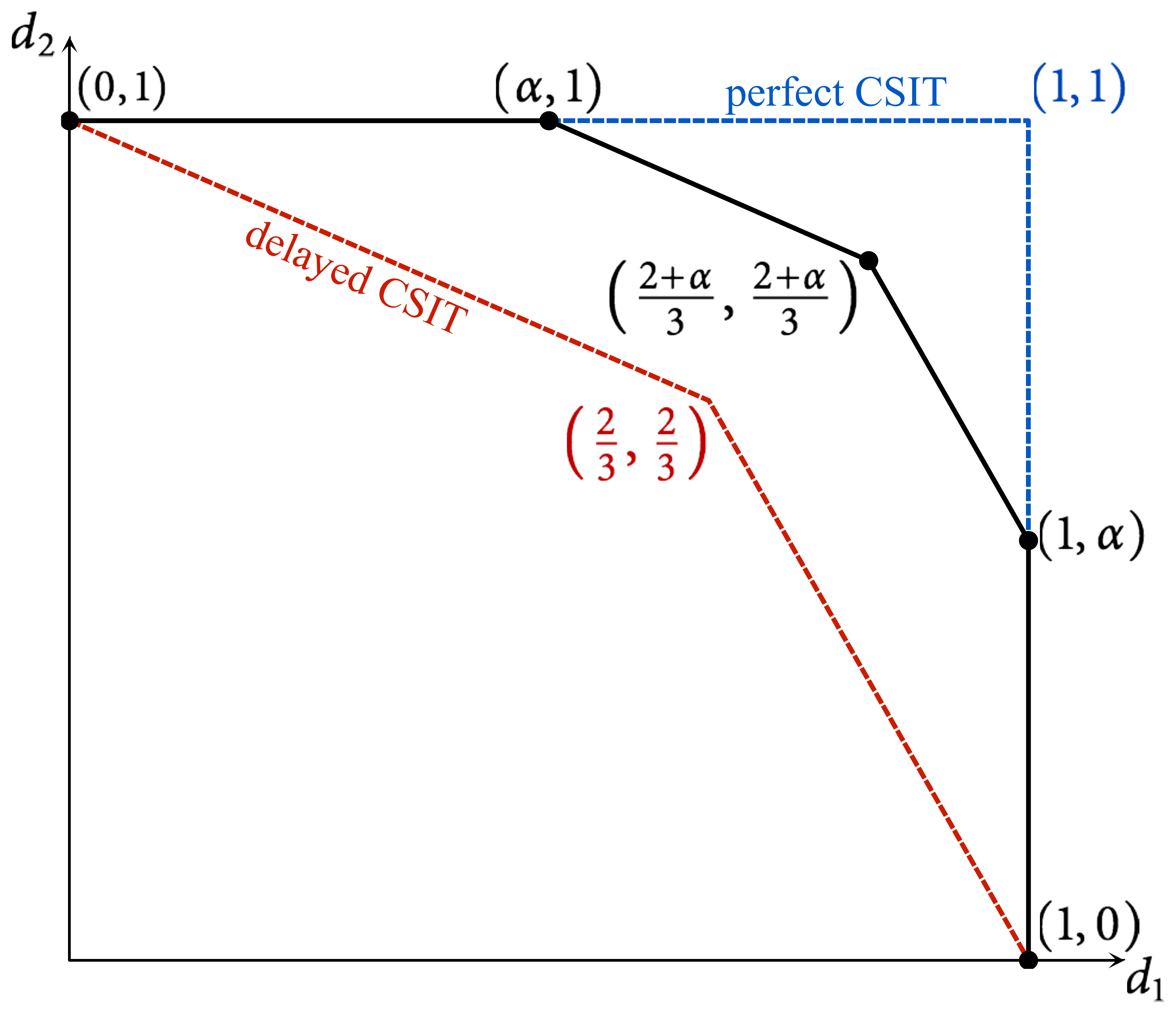}
\caption{DoF region of a two-user MISO channel with perfect delayed and
imperfect current CSI at the transmitter. The estimation error of the
current state scales as $P^{-\alpha}$.}
\label{fig:DoF-Region}
\end{figure}
\vspace{-\baselineskip}
As shown in Fig.~\ref{fig:DoF-Region}, the DoF region is a polygon characterized
by the vertices: $(0,1)$, $(\alpha,1)$, $(\frac{2+\alpha}{3},
\frac{2+\alpha}{3})$, $(1,\alpha)$, $(1,0)$.  
Note that the region collapses to the MAT
region~\cite{maddah2010degrees} when the quality of current CSIT is
poor~($\alpha\to0$), whereas it grows smoothly towards the DoF region
with perfect CSIT when $\alpha$ increases. In the following sections, we start 
with the converse proof by establishing outer bounds. Then, we propose
schemes that achieve the corner points of the region.

\section{Converse}
\label{sec:converse}

In this section, we establish the converse proof of the main result.
Before going into the details, we would like to point out the essential
elements of the upcoming proof: 
\begin{itemize}
  \item Genie-aided model: construct a degraded broadcast
    channel, as in \cite{maddah2010degrees}.
  \item Extremal inequality: bound the weighted difference
    of differential entropies~\cite{LiuViswanath}.
  \item Isotropic property of the channel uncertainty: tight upper bound on
    the pre-log factor. 
\end{itemize}

First, let us consider the 
genie-aided model where the genie provides the received signal $\{z_t\}$
of user~2 to user 1. 
This is a degraded broadcast channel $X \leftrightarrow (Y,Z)
\leftrightarrow Z$. Therefore, we
have the following upper bounds on the rates~$(R_1, R_2)$:
\begin{align}
  n R_1  &\le H(W_1) \\
  &= H(W_1 \cond S^n, \hat{S}^n) \\
  &= I(W_1; Y^n, Z^n \cond S^n, \hat{S}^n) + n \epsilon_n
  \label{eq:Fano1} \\
  &\le I(W_1; Y^n, Z^n, W_2 \cond S^n, \hat{S}^n) + n \epsilon_n \\
  &=  I(W_1; Y^n, Z^n \cond S^n, \hat{S}^n, W_2) + n \epsilon_n \\
  &= \sum_{i=1}^n I(W_1; Y_i, Z_i \cond Y^{i-1}, Z^{i-1}, S^n, \hat{S}^n, W_2) + n \epsilon_n \\
  &\le \sum_{i=1}^n I(X_i; Y_i, Z_i \cond Y^{i-1}, Z^{i-1}, S^n,
  \hat{S}^n, W_2) + n \epsilon_n \ \  \label{eq:dataproc} \\
  &= \sum_{i=1}^n I(X_i; Y_i, Z_i \cond Y^{i-1}, Z^{i-1}, S^i,
  \hat{S}^i, W_2) + n \epsilon_n \label{eq:future1}\\
  &= \sum_{i=1}^n \bigl( h(Y_i, Z_i \cond Y^{i-1}, Z^{i-1}, S^i, \hat{S}^i,
  W_2)  \\ 
  &\qquad  - h(Y_i, Z_i \cond X_i, Y^{i-1}, Z^{i-1}, S^i, \hat{S}^i, W_2)
  \bigr) + n \epsilon_n \\
  &= \sum_{i=1}^n \bigl( h(Y_i, Z_i \cond T_i, S_i) - h(\Noisey_i,
  \Noisez_i) \bigr) + n \epsilon_n \\
  &\le \sum_{i=1}^n h(Y_i, Z_i \cond T_i, S_i) + n \epsilon_n
  \label{eq:conv1}
\end{align}%
\begin{align} 
  n R_2 & \le H(W_2) \\
  &\le I(W_2; Z^n \cond S^n, \hat{S}^n) + n \epsilon_n   \label{eq:Fano2}\\
  &= \sum_{i=1}^n I(W_2; Z_i \cond Z^{i-1}, S^i, \hat{S}^i) + n
  \epsilon_n  \label{eq:future2}\\
  &= \sum_{i=1}^n \bigl( h(Z_i \cond Z^{i-1}, S^i, \hat{S}^i)\\
  &\qquad - h(Z_i \cond Z^{i-1}, S^i, \hat{S}^i, W_2) \bigr) + n
  \epsilon_n \\
  &\le \sum_{i=1}^n \bigl( h(Z_i \cond S_i) \\ 
  &\qquad - h(Z_i \cond Y^{i-1},
  Z^{i-1}, S^i, \hat{S}^i, W_2) \bigr) + n \epsilon_n \label{eq:entropy} \\
  &= \sum_{i=1}^n \bigl( h(Z_i \cond S_i) - h(Z_i \cond T_i, S_i)
  \bigr) + n \epsilon_n
  \label{eq:R2}
\end{align}%
where we define $T_i \defeq (Y^{i-1}, Z^{i-1}, S^{i-1}, \hat{S}^{i},
W_2)$. Note that the above chains of
inequalities follow closely Gallager's proof for the degraded broadcast
channel~\cite{Gallager_BC}~(also see \cite{Cover_Thomas}), with the
integration of the channel states. In particular, \eqref{eq:Fano1} and
\eqref{eq:Fano2} are from Fano's inequality; \eqref{eq:dataproc} is from
the data processing inequality; \eqref{eq:future1} holds because the
input $X_i$ and the outputs $(Y_i, Z_i)$ of the channel at
instant $i$ do not depend on the future states given the past and
current states; \eqref{eq:future2} results from the same reasoning and the
chain rule of mutual information; \eqref{eq:conv1} is from the
non-negativity of the differential entropy of unit-variance AWGN, i.e.,
$h(\Noisey_i, \Noisez_i)\ge0$; \eqref{eq:entropy} holds since
removing~(resp.~adding) conditions does not decrease~(resp.~increase) differential entropy.   
In the following, we would like to obtain an upper bound on $R_1+2R_2$. From
\eqref{eq:conv1} and \eqref{eq:R2}, we have
\begin{align}
  n(R_1 + 2R_2) &\le \sum_{i=1}^n \bigl(  2 h(Z_i \cond S_i) + h(Y_i,
  Z_i \cond T_i, S_i) \\
  & \qquad - 2 h(Z_i \cond T_i, S_i) \bigr) + 3 n \epsilon_n. 
  \label{eq:R1+2R2}
\end{align}%
Now, we can upper-bound each term in the above summation: 
\begin{align}
\MoveEqLeft[1] 2 h(Z_i \cond S_i) + h(Y_i, Z_i \cond T_i, S_i) - 2 h(Z_i \cond T_i, S_i)  
\nonumber \\
&\le \max_{P_{T_i} P_{X_i\vert T_i} } \!\!\bigl( 2 h(Z_i \cond S_i) \\
&\qquad + h(Y_i, Z_i \cond T_i, S_i) - 2
h(Z_i \cond T_i, S_i) \bigr) \\
&\le \max_{P_{T_i} P_{X_i\vert T_i} } \!\! 2 h(Z_i \cond
S_i) \\
&\qquad + \max_{P_{T_i} P_{X_i\vert T_i} } \!\! \bigl(h(Y_i, Z_i \cond T_i, S_i) -
2 h(Z_i \cond T_i, S_i)\bigr).\ \quad \label{eq:tmp776} 
\end{align}%
The first maximization can be upper-bounded as: 
\begin{align}
   \max_{P_{T_i} P_{X_i\vert T_i} } 2 h(Z_i \cond S_i)
   &\le 2\, \E_{G_i} \Bigl( \max_{P_{X_i\vert G_i= \gv_i}} h(\gv_i^\H
   X_i  + \Noisey_i )  \Bigr) \\ 
   &\le 2\, \E_{G_i} \bigl( \log(1+P \Norm{\gv_i})  \bigr) \\
   &\le 2\log P + O(1) \label{eq:hZub} 
\end{align}%
where, to get the first inequality, we put the maximization into the
expectation; the second inequality is from the fact that Gaussian distribution 
maximizes differential entropy under the covariance constraint, that the
logarithmic function is monotonically increasing, and that the
following partial ordering holds $\mathsf{Cov}(X_i \cond \gv_i)
\preceq \mathsf{Cov}(X_i) \preceq P \Id$; the last one is from
Jensen's inequality. The second maximization in
\eqref{eq:tmp776} can also be bounded, but in a slightly more involved
way, as shown in \eqref{eq:tmp20}-\eqref{eq:tmp26} on the top of next page. 
\begin{figure*}[!t]
\normalsize
\begin{align}
  \MoveEqLeft
  \max_{P_{T_i} P_{X_i\vert T_i}} \bigl(h(Y_i, Z_i \cond T_i, S_i) - 2 h(Z_i \cond T_i, S_i)\bigr) \\
&\le \max_{P_{T_i}} \E_{T_i}\Bigl( \max_{P_{X_i\vert T_i} } \bigl(h(Y_i, Z_i \cond T_i=T, S_i) - 2 h(Z_i \cond T_i=T, S_i)\bigr) \Bigr)
\label{eq:tmp20}
 \\
 &= \max_{P_{T_i}} \E_{T_i}\Bigl( \max_{P_{X_i\vert T_i}}
 \E_{S_i \vert T_i} \bigl(  h(Y_i, Z_i \cond
 T_i=T, S_i = \Sm_i) - 2 h(Z_i \cond T_i=T, S_i=\Sm_i) \bigr)\Bigr) 
 \label{eq:tmp21}
 \\
 &= \max_{P_{T_i}} \E_{T_i}\Bigl( \max_{P_{X_i \vert T_i}} \E_{S_i \vert
 \hat{S}_i} \bigl(  h(\Sm_i X_i + N_i \cond T_i=T) - 2 h( \gv_i^\H X_i + E_i
 \cond T_i=T) \bigr)\Bigr) 
 \label{eq:tmp22}
 \\
 &= \max_{P_{T_i}} \E_{T_i}\biggl( \max_{\Cm:\Cm\succeq0,\trace[\Cm]\le P}
 \max_{P_{X_i \vert T_i}: \atop \mathsf{Cov}(X_i\vert T_i) \preceq \Cm}
 \E_{S_i \vert \hat{S}_i} \bigl(  h(\Sm_i X_i + N_i \cond T_i=T) - 2 h( \gv_i^\H X_i + E_i \cond T_i=T) \bigr) \biggr)
 \label{eq:tmp23}
 \\
 &= \max_{P_{T_i}} \E_{T_i}\biggl( \max_{\Cm:\Cm\succeq0,\trace[\Cm]\le P}
 \E_{S_i \vert \hat{S}_i} \bigl(  \log\det[\Id+\Sm_i \Km_* \Sm_i^\H]  -
 2\log(1+\gv_i^\H \Km_* \gv_i) \bigr) \biggr) 
 \label{eq:tmp24}
 \\
 &\le  \E_{\hat{S}_i}\biggl( \max_{\Km:\Km\succeq0,\trace[\Km]\le P}
 \E_{S_i \vert \hat{S}_i} \bigl(  \log\det[\Id+\Sm_i \Km \Sm_i^\H] 
  - 2\log(1+\gv_i^\H \Km \gv_i) \bigr)\biggr) 
 \label{eq:tmp25} \\
 &\le  \E_{\hat{S}_i}\biggl( \max_{\Km:\Km\succeq0,\trace[\Km]\le P}
 \E_{S_i \vert \hat{S}_i} \bigl(  \log(1+\hv_i^\H \Km \hv_i) 
  - \log(1+\gv_i^\H \Km \gv_i) \bigr)\biggr) 
 \label{eq:tmp26}
\end{align}
\hrulefill
\vspace*{4pt}
\end{figure*}
We get \eqref{eq:tmp20} by putting one of the maximizations into the
expectation, which does not decrease the value; in \eqref{eq:tmp22}, we
define $N_i \defeq [\Noisey_i\ \ \Noisez_i]^\T$; \eqref{eq:tmp23}
is obtained by splitting one maximization into two, one with
the trace constraint and the other with the covariance constraint;
\eqref{eq:tmp24} is from the fact that with covariance constraint,
Gaussian distribution maximizes the weighted difference of two
differential entropies, given that~i)~$S_i$ is independent of $X_i$
conditional on $T_i = (Y^{i-1}, Z^{i-1}, S^{i-1}, \hat{S}^{i},
W_2)$ due to the Markovian~\eqref{eq:Markov} and the fact that
$X_i$ is a function of the messages $(W_1, W_2)$, the past states
$S^{i-1}$, and the estimates up to the current state $\hat{S}^i$, and
that ii)~$Y_i$ is a degraded version of
$(Y_i, Z_i)$; this is an application of the extremal inequality
\cite{LiuViswanath,Weingarten}; note that
$\Km_*\preceq\Cm$ is defined as the optimal covariance for the inner
maximization; \eqref{eq:tmp25} holds because any
$\Km$ such that $0 \preceq \Km \preceq \Cm$ with $\trace[\Cm]\le P$
belongs to the set $\left\{ \Km:\ \Km\succeq0, \trace[\Km]\le P
\right\}$, and that the whole term only depends on $\hat{S}_i$; the last
inequality is from the fact that $\det[\Id+\Am] \le
(1+a_{11})(1+a_{22})$ for any $\Am\defeq[a_{ij}]_{i,j=1,2} \succeq
\mbs{0}$. 

\begin{lemma}\label{lemma:logdet}
  For any given $\Km\succeq0$ with eigenvalues $\lambda_1\ge\cdots\ge\lambda_m\ge0$,
  we have 
  \begin{align}
\E_{S_i \vert \hat{S}_i}\! \bigl( \log(1+\hv_i^\H \Km \hv_i) \bigr)
&\le \log(1+\Norm{\hat{\hv}_i} \lambda_1) + O(1), \ \quad \label{eq:tmp89} \\
\E_{S_i \vert \hat{S}_i}\! \bigl(  \log(1+\gv_i^\H \Km \gv_i) \bigr)
&\ge  \log(1+ 2^{\gamma} \sigma^2 \lambda_1 ) 
+ O(1), \label{eq:tmp99} 
  \end{align}%
  with
  \begin{equation}
    \gamma \defeq \E_{\tilde{S}_{i} \vert
  \hat{S}_i} \Bigl( \log \frac{\Abs{\tilde{g}_{i,1}}}{\sigma^2}
  \Bigr).
  \end{equation}
\end{lemma}
\vspace{\baselineskip}
\begin{proof}
  See Appendix~\ref{app:proof-lemma-logdet}.
\end{proof}
It is worth noting that $\gamma$ is finite according to
Assumption~\ref{assumption:fading}. Therefore, $2^\gamma$ is a
strictly positive and bounded value that can be regarded as constant 
as far as the DoF is concerned. From Lemma~\ref{lemma:logdet}, we have
\begin{align}
\MoveEqLeft 
\E_{S_i \vert \hat{S}_i} \bigl( \log(1+\hv_i^\H \Km \hv_i)
- \log(1+\gv_i^\H \Km \gv_i) \bigr) \nonumber \\ &\le
\log\frac{1+\Norm{\hat{\hv}_i} \lambda_1}{1+ 2^{\gamma} \sigma^2
\lambda_1} + O(1) \label{eq:tmp667}
 \\
&\le \log\biggl( 1 + \frac{\Norm{\hat{\hv}_i}}{2^{\gamma} \sigma^2}
\biggr) + O(1)\label{eq:tmp666}
\\
&\le - \log(\sigma^2) + \log\bigl( 2^{\gamma}\sigma^2 +
{\Norm{\hat{\hv}_i}} \bigr)  + O(1) \label{eq:tmpub}
\end{align}%
where \eqref{eq:tmp666} is from the fact that $\log\frac{1+ax}{1+bx} \le
\log(1+\frac{a}{b})$, $\forall\, a,x\ge0,\,b>0$.  
Note that the above upper bound does not depend on $\Km$. From
\eqref{eq:tmp26} and \eqref{eq:tmpub} and by noticing
that $\sigma^2\le1$, we have 
\begin{align}
    \MoveEqLeft \max_{P_{T_i} P_{X_i\vert T_i}} \bigl(h(Y_i, Z_i \cond T_i, S_i) -
2 h(Z_i \cond T_i, S_i)\bigr)  \nonumber \\
&\le \alpha \log P + \E_{\hat{S}_i}\bigl( \log\bigl( 2^{\gamma} +
{\Norm{\hat{\hv}_i}} \bigr) \bigr) + O(1) \\
&= \alpha \log P + O(1). \label{eq:hYZ-2hZ}
\end{align}%
From \eqref{eq:R1+2R2}, \eqref{eq:hZub}, \eqref{eq:hYZ-2hZ}, and by letting
$n \to \infty$, we have
\begin{align}
  R_1 + 2R_2 &\le (2+\alpha) \log P + O(1), 
\end{align}%
from which we obtain \eqref{eq:d1+2d2} by dividing both sides of the
above inequality by $\log P$ and tending $P\to\infty$. Similarly, from
\eqref{eq:R2} and \eqref{eq:hZub}, and by letting $n\to\infty$, we have
\begin{align}
  R_2 &\le \log P + O(1),
\end{align}%
from which the single user bound \eqref{eq:d2} follows immediately. To obtain
\eqref{eq:d1} and \eqref{eq:2d1+d2}, we can use the genie-aided model in which receiver~2
is helped by the genie and has perfect knowledge of $y_t$. Due to the
symmetry, the same reasoning as above can be applied by swapping the
roles of receiver~1 and receiver~2. The converse part is thus completed. 

\begin{remark}
In a nutshell, the converse proof can be summarized as follows, in terms
of the essential elements mentioned at the beginning of this section.
First, the ``degraded'' property enables the use of the extremal
inequality~(cf.~\eqref{eq:tmp23} and \eqref{eq:tmp24}).  Then, the
latter provides a closed-form upper bound given by the Gaussian
distribution~(cf.~\eqref{eq:tmp26}). Finally, the isotropic property of
the channel uncertainty is exploited only at the end of the proof, to
bound the expectation of the logarithmic function~(cf.~\eqref{eq:tmp99}). 
\end{remark}

\section{Achievability}

To show the achievability of the whole region, it is enough to show that
all corner points in Fig.~\ref{fig:DoF-Region} are achievable. Note
that the extreme points $(1,0)$ and $(0,1)$ can be trivially achieved by serving
only one of the users.  The rest of the section is devoted to proving the
achievability of $(1,\alpha)$, $(\alpha,1)$, and $\left(
\frac{2+\alpha}{3}, \frac{2+\alpha}{3} \right)$. 
Since the DoF region does not depend on the number of transmit
antennas~$m$, $\forall\,m\ge2$, it is enough to prove the achievability for the case
$m=2$ which is assumed implicitly in this section. 
The exact achievable rate region from which the
DoF can be derived in a more rigorous way is provided in the 
appendix.

\subsection{Achieving $(1,\alpha)$ and $(\alpha,1)$}

One of the key elements to achieve the three corner points is broadcasting
with common message in the presence of imperfect current CSIT. The
following result is crucial and will be repeatedly used in the proofs. 

\begin{lemma}[broadcast channel with common message]
  \label{lemma:BC-CM}
  Let $(R_\text{c}, R_{\text{p}1}, R_{\text{p}2})$ be the rate of common message, private message for
user~1, and private message for user~2, respectively. Furthermore, we
let $(d_\text{c}, d_{\text{p}1}, d_{\text{p}2})$ be the corresponding DoF. Then, there exists a
family of codes~$\left\{ \Xc_\text{c}(P), \Xc_{\text{p}1}(P), \Xc_{\text{p}2}(P) \right\}$, such that 
\begin{align}
  d_\text{c} &= 1-\alpha, \quad \text{and} \quad
  d_{\text{p}1} = d_{\text{p}2} = \alpha
\end{align}%
are achievable simultaneously. 
\end{lemma}
  A sketch of proof is as follows, with more details given in
  Appendix~\ref{app:BC-CM}. Let us consider a single
  channel use with a superposition scheme: $\xv = \xv_c + \xv_{\text{p}1} +
  \xv_{\text{p}2}$ with precoding such that $\E[\xv_{\text{p}1}\xv_{\text{p}1}^\H] =
  \frac{P_\text{p}}{2} \Psimgp$ and
  $\E[\xv_{\text{p}2}\xv_{\text{p}2}^\H] = \frac{P_\text{p}}{2} \Psimhp$.
  We set the power $P_\text{p} \sim P^{\alpha}$ such
  that the private signals are drowned by the AWGN at the unintended
  receivers while remaining the level $P^{\alpha}$ at the intended
  receivers. The power of the common signal is $P_\text{c} = \E[\Norm{\xv_c}] \sim P$.
  The decoding is performed as follows. At each receiver, the common
  message is decoded first by treating the private signals as noise.
  The signal-to-interference-and-noise ratio~(SINR) is approximately
  $P_\text{c} / P_\text{p} \sim P^{1-\alpha}$, from which the
  achievability of $d_\text{c} = 1-\alpha$ is
  shown. Then, each receiver decodes their own private
  messages, after removing the decoded common message. The SINR for the
  private message being approximately $P^{\alpha}$, $d_{\text{p}k} = \alpha$ is
  thus achievable for user~$k$, $k=1,2$.  

From the above lemma, the achievability of $(1,\alpha)$ is 
straightforward. Let $W_1$ and $W_2$ be the messages for user~1 and
user~2, respectively. Assuming that the DoF are respectively
$d_1$ and $d_2$, we can split user~1's message as $W_1 = (W_{10},
{W}_{11})$
with the corresponding rate-splitting $d_1 = d_{10} + {d}_{11}$. 
Then, $(W_{10}, {W}_{11}, W_2)$ are broadcast to both users with
$W_{10}$ as common message. According to Lemma~\ref{lemma:BC-CM},
$(W_{10},{W}_{11})$ and
$(W_{10}, W_2)$ can be recovered by user~1 and user~2, respectively, as
long as
\begin{align}
  d_{10} &\le 1-\alpha, \quad {d}_{11} \le \alpha, \quad \text{and}\quad d_2 \le \alpha
\end{align}%
which implies $d_1 = d_{10} + {d}_{11} \le 1$ and $d_2 \le \alpha$ are achievable
simultaneously. Similarly, $(\alpha, 1)$ can also be achieved by
the same scheme with rate-splitting over user~2's message. 

The proposed scheme, hereafter referred to as rate-splitting~(RS),
achieves both corner points $(1,\alpha)$ and $(\alpha,1)$ with only
current CSIT 
and without delayed CSIT at all. A sum DoF of $1+\alpha$ is thus
attained. The idea is
closely related to the Han-Kobayashi scheme~\cite{Han-Kobayashi} for
the two-user interference channel where each receiver can decode and
then eliminate the common part of the interfering signal to achieve a
higher rate. Therefore, the common message in our RS scheme is desirable
for only one of the users but is decodable by both users.


\subsection{Achieving the symmetric corner point $\left( \frac{2+\alpha}{3}, \frac{2+\alpha}{3} \right)$}

In the following, we show that exploiting both current and delayed CSIT,
the symmetric corner point $\left( \frac{2+\alpha}{3}, \frac{2+\alpha}{3} \right)$
can be achieved. It provides a sum DoF of $\frac{2(2+\alpha)}{3}$
that is strictly larger than $1+\alpha$ for $\alpha<1$.   
Since this scheme builds on the MAT scheme, we briefly review it first. 

\subsubsection{MAT alignment revisited}

In the two-user MISO case, the original MAT is a three-slot
scheme, described by the equations
\begin{align}
  \xv_1 &= \uv & 
  \xv_2 &= \vv & 
  \xv_3 &= [\,\gv_1^\H \uv + \hv_2^\H \vv  \quad 0\,]^\T
  \label{eq:tmp201}\\
  y_1 &= \hv_1^\H \uv & 
  y_2 &= \hv_2^\H \vv & 
  y_3 &= h_{31}^* (\gv_1^\H \uv + \hv_2^\H \vv)  \\
  z_1 &= \gv_1^\H \uv &
  z_2 &= \gv_2^\H \vv &
  z_3 &= g_{31}^*(\gv_1^\H \uv + \hv_2^\H \vv) 
  \label{eq:tmp203} 
\end{align}
where $\xv_t\in\mathbb{C}^{m\times1}, y_t, z_t\in\mathbb{C}$ are the
transmitted signal, received signals at user~1 and user~2, respectively, at time slot $t$;
$\uv,\vv\in\mathbb{C}^{m\times 1}$ are useful signals to user~1 and
user~2, respectively; for simplicity, we omit the noise in the received
signals.  
The idea of the MAT scheme is to use delayed CSIT to align the mutual
interference into a one-dimensional subspace~($\hv_1^\H
\vv$ for user 1 and $\gv_1^\H \uv$ for user 2). And importantly, the
interference is reduced without sacrificing the dimension of
the useful signals. 
Specifically, a two-dimensional interference-free observation
of $\uv$~(resp. $\vv$) is obtained at receiver~1~(resp.~receiver~2).

Interestingly, the alignment can be done in a different manner.
\begin{flalign}
   \xv_1 &= \uv + \vv& 
   \xv_2 &= [\,\hv_1^\H \vv \ \ 0\, ]^\T & 
   \xv_3 &= [\,\gv_1^\H \uv \ \ 0\, ]^\T & 
   \\
   y_1 &= \hv_1^\H (\uv + \vv) & 
   y_2 &= h_{21}^*\hv_1^\H \vv & 
   y_3 &= h_{31}^*\gv_1^\H \uv \\
   z_1 &= \gv_1^\H (\uv + \vv) &
   z_2 &= g_{21}^*\hv_1^\H \vv & 
   z_3 &= g_{31}^*\gv_1^\H \uv 
\end{flalign}
In the first slot, the transmitter sends the private signals to both
users by simply superposing them. In the second slot, the transmitter
sends the interference overheard by receiver 1 in the first slot. The
role of this stage is two-fold: \emph{resolving interference for user~1
and reinforcing signal for user~2}. In the third slot, the transmitter
sends the interference overheard by user 2 to help both users the other
way around. In summary, this variant of the MAT scheme consists of two phases:
i)~broadcast of the private signals, and ii)~multicast of the overheard
interferences. 
At the end of three time slots, the observations at the receivers are given by
\begin{align}
  \begin{bmatrix} y_1\\y_2\\y_3 \end{bmatrix} &= \underbrace{\begin{bmatrix}
    \hv_1^\H \\ 0 \\ h_{31}^* \gv_1^\H 
  \end{bmatrix}}_{\text{rank}=2} \uv + \underbrace{\begin{bmatrix}
    \hv_1^\H \\ h_{21}^* \hv_1^\H \\ 0 
  \end{bmatrix}}_{\text{rank}=1} \vv, 
\shortintertext{and} 
 \begin{bmatrix} z_1\\z_2\\z_3 \end{bmatrix} &= \underbrace{\begin{bmatrix}
   \gv_1^\H \\ g_{21}^* \hv_1^\H \\ 0   
 \end{bmatrix}}_{\text{rank}=2} \vv + \underbrace{\begin{bmatrix}
   \gv_1^\H \\  0 \\ g_{31}^* \gv_1^\H  
 \end{bmatrix}}_{\text{rank}=1} \uv. 
\end{align}%
For each user, the useful signal lies in a two-dimensional subspace
while the interference is aligned in a one-dimensional subspace.  It
readily follows that this variant enables each user to achieve two
degrees of freedom in the three-dimensional time space as for the
original MAT scheme.  Although the original and variant schemes are equivalent 
from the point of the
space-time alignment, they differ conceptually in the way how the
``order-two'' symbols are delivered. More precisely, the variant spends
two slots to deliver two separate symbols: the interferences overheard
by user~1 and user~2, denoted by 
\begin{align}
  \eta_1 &\defeq \hv_1^\H\vv \quad \text{and} \quad \eta_2 \defeq \gv_1^\H \uv,
\end{align}%
while the original MAT spends a single slot to deliver one symbol
$\hv_2^\H\vv+\gv_1^\H \uv$.


\subsubsection{Proposed scheme}
\label{sec:proposed-scheme}

\begin{figure*}
 \centering
\includegraphics[width=0.75\textwidth]{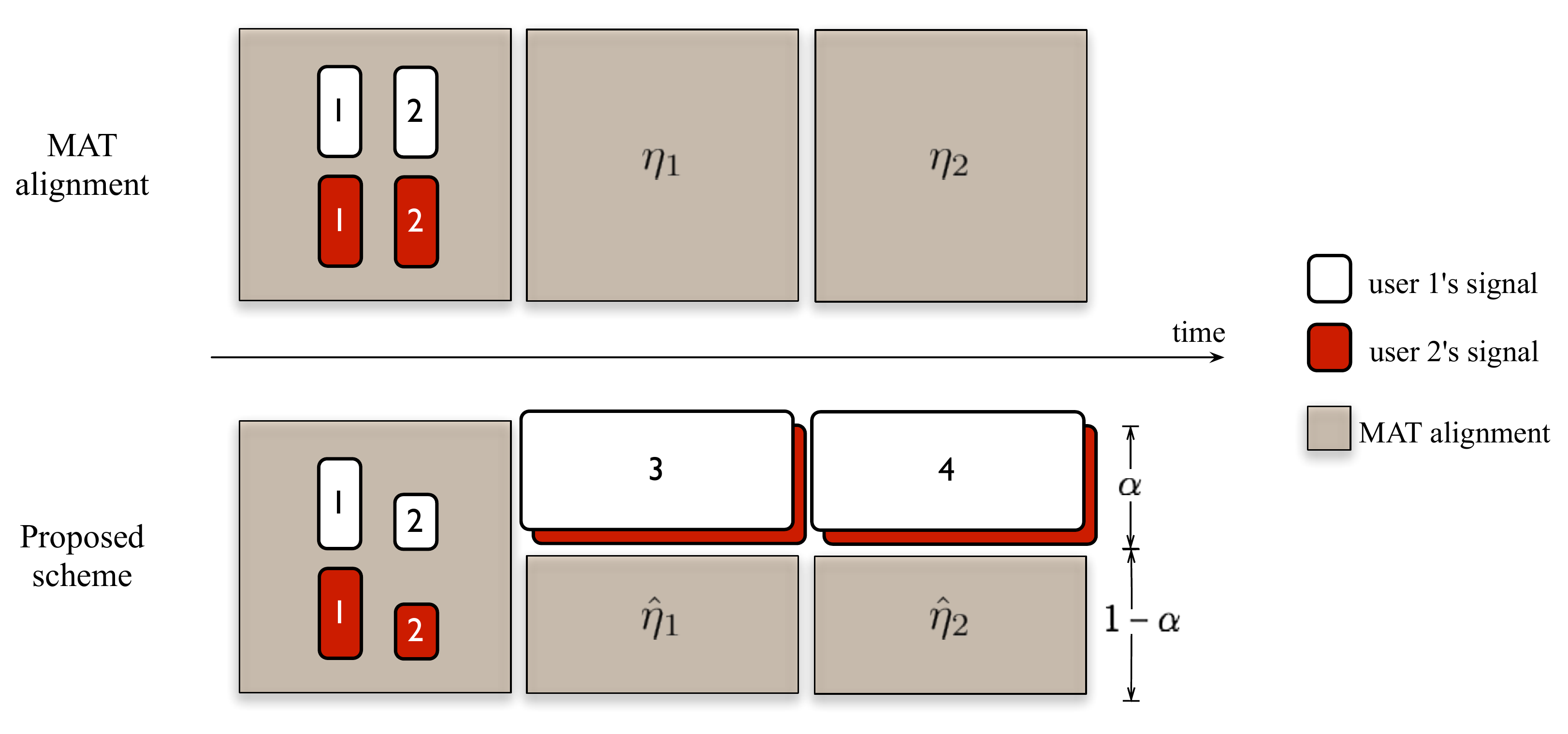} 
\caption{Overview of the main differences between the Maddah-Ali Tse
alignment and the proposed scheme.}
\label{fig:proposed-vs-MAT}
\end{figure*}

Based on the above variant of the MAT alignment, we propose a new scheme
that exploits optimally both the perfect delayed and imperfect current
CSIT. Before proceeding further, we would like to highlight
the main ideas as compared to the MAT alignment~(Fig.~\ref{fig:proposed-vs-MAT}):
\begin{itemize}
  \item Spatial precoding and power allocation in the first slot: $1+(1-\alpha)
    = 2-\alpha$ instead of two streams are broadcast.
  \item Digitizing the overheard interferences $(\eta_1, \eta_2)$ in
    approximately $2(1-\alpha) \log P$ bits.
  \item Broadcasting the digitized interferences $(\hat{\eta}_1,
    \hat{\eta}_2)$ as common message \emph{and} two new private messages
    of $\alpha \log P$ bits each, in the second and third slots. 
\end{itemize}
These ideas will be explored in the rest of the section whereafter the
interpretation of Fig.~\ref{fig:proposed-vs-MAT} will become clear.  Since only $\hv_1$ and $\gv_1$ are involved below, we drop the time indices for convenience.

\subsubsection*{\textbf{Spatial precoding and power allocation}}

As in the MAT alignment, we first superpose the two private signals
as $\xv = \uv + \vv$, except that $\uv$ and $\vv$ are precoded
beforehand. The precoding is specified by the covariance matrices 
\begin{align}
  \Qu \defeq \E[\uv \uv^\H] \quad  \text{and} \quad \Qv \defeq \E[\vv \vv^\H]
\end{align}%
that may depend on the estimates of the current channel. The power constraint is
respected by choosing $\Qu$ and $\Qv$ such that
$\trace[\Qu] + \trace[\Qv] \le P$. 
In particular, we choose $\Qu$ and $\Qv$ in such a way that the
power of the interferences $\eta_1$ and $\eta_2$ is reduced and scales
as $O(P^{1-\alpha})$. To this end,  
\begin{itemize}
  \item for user~$k$, $k=1,2$, we send two streams of messages
    $(W_{k,1}, W_{k,2})$ in two orthogonal directions:
    one perpendicular to the estimated channel of the unintended user, while
    the other one aligned with it, i.e., 
    \begin{align}
      \Qu &= P_1 \Psimgp + P_2 \Psimg, \label{eq:Qu} \\
      \Qv &= P_1 \Psimhp + P_2 \Psimh; 
    \end{align}%
  \item the transmit power in the estimated channel direction is
    such that $P_2 \sim P^{1-\alpha}$, whereas the transmit power in the
    orthogonal direction is $P_1 = P - P_2 \sim P$ for any $\alpha<1$.  
\end{itemize}
With $\Qu$ and $\Qv$ chosen as such, it is readily shown that, for a
given channel realization $\hv$, the power of the interference seen by user~1 is 
\begin{align}
  \sigma_{\eta_1}^2 &\defeq \E_{\vv}(\abs{\hv^\H \vv}^2) \\
  &= \hv^\H \Qv \hv \\
  &= \hv^\H (P_1 \Psimhp + P_2 \Psimh) \hv \\
  &= P_1 \tilde{\hv}^\H \Psimhp \tilde{\hv} + P_2 {\hv}^\H \Psimh {\hv}
  \\ 
  &\le P_1 \norm{\tilde{\hv}}^2 + P_2 \norm{{\hv}}^2.   
\end{align}%
By averaging $\sigma_{\eta_1}^2$ over $\hv$, we have
\begin{align}
 \E[\sigma_{\eta_1}^2] = O(P^{1-\alpha}).
\label{eq:averageINT}
\end{align}%
Due to the symmetry, defining $\sigma_{\eta_2}^2 \defeq
\E_{\uv}(\abs{\gv^\H \uv}^2)$, we
also have $\E[\sigma_{\eta_2}^2] = O(P^{1-\alpha})$.

\subsubsection*{\textbf{Digitizing the overheard interferences
}} 

As in the second phase of the MAT variant, we would like to convey the overheard 
interferences $(\hv^\H\vv, \gv^\H \uv)$ to both receivers. However, unlike the original MAT
scheme where these symbols are transmitted in an
analog fashion, we quantize them and then transmit the digital version.
The rationale behind this choice is as follows. With the precoding and
power allocation as described above, the overheard
interferences have a reduced power~$O(P^{1-\alpha})$, without sacrificing \emph{too much}
received signal power.\footnote{With no CSIT on the current channel, the
only way to reduce the interference power is to reduce the transmit
power, therefore the received signal power.} As a result, we should be
able to compress the interferences, which in turn makes room for
transmission of new symbols. The benefit can be significant when the current
CSIT is nearly perfect. In this case, the analog transmission is no
longer suitable, due to the mismatch between the source (interference)
power and available transmit power. Therefore, a good alternative is to
quantize the interferences and to transmit the encoded symbols. The
number of quantization bits depends naturally on the interference power
that is related to the quality of the current channel state information.  


For simplicity, we suppose that $\eta_1$ and $\eta_2$ are quantized separately. 
Furthermore, let us assume that an $R_{\eta_k}$-bits quantizer is used for $\eta_k$, $k=1,2$.
Hence, we have  
\begin{align}
  \eta_k &= \hat{\eta}_k + \Delta_k 
\end{align}%
where $\hat{\eta}_k$ and ${\Delta_k}$ are respectively the quantized value and the
quantization noise with average distortion $\E[\Abs{\Delta_k}] = D_k$,
$k = 1,2$. 
The index corresponding to $\hat{\etav} \defeq (\hat{\eta}_1, \hat{\eta}_2)$,
represented in $\Retav \defeq R_{\eta_1}+R_{\eta_2}$ bits, is then multicast to both
users. In order not to incur a DoF loss with the quantization, we set the
distortion to the noise level, i.e., $D_1 = D_2 = 1$.  
With the above choices, we can upper-bound the quantization rate $\Retav$ 
\begin{align}
  \Retav &\le \E \biggl( \log\biggl( \frac{\sigma_{\eta_1}^2}{D_1}
  \biggr)\biggr) +
  \E\biggl( \log\biggl( \frac{\sigma_{\eta_2}^2}{D_2} \biggr) \biggr) \\
  &\le \log\Bigl( \E[\sigma_{\eta_1}^2] \Bigr) + \log\Bigl(
  \E[\sigma_{\eta_2}^2] \Bigr) \\
&\le 2 (1-\alpha) \log P + O(1) 
\end{align}%
where the first inequality is from the rate-distortion theorem and the
fact that Gaussian source is the hardest to
compress~\cite{Cover_Thomas}; the second inequality is from the concavity of the
log function and Jensen's inequality; the last one is from \eqref{eq:averageINT}. 

\subsubsection*{\textbf{Multicasting digitized interferences and broadcasting
new private messages}}

The next step is to communicate the digitized interferences~$(\hat{\eta}_1,
\hat{\eta}_2)$, represented approximately in $2(1-\alpha) \log P$ bits, to both users. 
This information is broadcast as common
message in two slots. Meanwhile, new private messages $(W_{1,3}, W_{2,3})$
and $(W_{1,4}, W_{2,4})$ are sent to both users simultaneously in the
second and third slots, respectively. The superposition is illustrated
in Fig.~\ref{fig:proposed-vs-MAT}. 
In the following, we let $(d_\text{c}, d_{\text{p}1}, d_{\text{p}2})$ denote the corresponding DoF per
slot for the common message, private messages for user~1 and user~2,
respectively. It is readily shown that $d_\text{c} = 1-\alpha$. 

\subsubsection*{Decoding}
Each user first decodes the second and
third slots, i.e., receiver~$k$ recovers $(\hat{\eta}_1, \hat{\eta}_2,
{W}_{k,3}, {W}_{k,4})$, $k=1,2$. According to Lemma~\ref{lemma:BC-CM}, and given that
$d_\text{c} = 1-\alpha$, these messages can be decoded reliably as long as 
\begin{align}  
  {d}_{\text{p}k} &\le \alpha, \quad k=1,2.
\end{align}%
Then, receiver~1 has the following equations
\begin{align}
  y &= \hv^\H \uv + \eta_1 + \noisey, \\
  \hat{\eta}_1 &= \eta_1 - {\Delta_1}, \\
  \hat{\eta}_2 &= \eta_2 - {\Delta_2} = \gv^\H \uv - {\Delta_2},
\end{align}%
from which an equivalent $2\times2$ MIMO channel is obtained
\begin{align}
	\tilde{\yv} \triangleq \begin{bmatrix} y -\hat{\eta}_1
          \\\hat{\eta}_2 \end{bmatrix}  = { \Sm } {\uv} +
            {\begin{bmatrix} \noisey + {\Delta_1} \\ -
              {\Delta_2} \end{bmatrix}} \label{eq:MIMO1}
\end{align}
where the noise $\bv \defeq [\noisey + {\Delta_1} \ -{\Delta_2}]^\T$
depends on the input signals in general. Similarly, receiver~2 has
\begin{align}
	\tilde{\zv} \triangleq \begin{bmatrix} \hat{\eta}_1 \\ z - \hat{\eta}_2
        \end{bmatrix}  = { \Sm } {\vv} +
            {\begin{bmatrix} - {\Delta_1} \\ \noisez + {\Delta_2}
            \end{bmatrix}}. 
              \label{eq:MIMO2}
\end{align}
In order to recover the messages $W_{\text{mimo},1} \defeq (W_{1,1}, W_{1,2})$ encoded in $\uv$ or $W_{\text{mimo},2} \defeq (W_{2,1}, W_{2,2})$ encoded in $\vv$, each user performs conventional MIMO
decoding of the above equivalent channel. Let $\Rmimo$ denote the achievable rate of the equivalent
channel~\eqref{eq:MIMO1} in bits per channel use and $\dmimo$ the
corresponding DoF. We can lower-bound $\Rmimo$ as follows: 
\begin{align}
  \Rmimo &= \E\bigl( I(U; \tilde{Y} \cond S = \Sm) \bigr) \\
  &= \E\bigl( I(\Sm U; \tilde{Y}) \bigr) \label{eq:invertible} \\
  &= \E\bigl( h(\Sm U) - h(\Sm U \cond \tilde{Y}) \bigr) \\
  &= \E\bigl( h(\Sm U) - h(\Noisey + \Delta_1, -\Delta_2 \cond \tilde{Y})
  \bigr)\\
  &\ge \E\bigl( h(\Sm U) - h(\Noisey + \Delta_1, -\Delta_2) \bigr)
  \label{eq:cond}\\
  &\geq \E\bigl( \log\det[\Sm \Qu \Sm^\H] \bigr) - \log(1+D_1) \\
  &\qquad - \log(D_2) \label{eq:ind}\\
  &= \E\bigl( \log\det[\Qu] \bigr)  + \E\bigl( \log\det[\Sm \Sm^\H]
  \bigr) \\
  &\qquad - \log(1+D_1) - \log(D_2)  \\
  &= \log(P_1 P_2) + O(1) \\
  &= (2-\alpha) \log(P) + O(1)
\end{align}%
where \eqref{eq:invertible} is from the fact that $\Sm$ is invertible
almost surely and therefore the linear transformation is
information-lossless; \eqref{eq:cond} holds since conditioning does not
increase differential
entropy; \eqref{eq:ind} follows because $\uv$ is Gaussian, then by
noticing that $E+\Delta_1$ and $\Delta_2$ are independent with the
corresponding differential entropies maximized by Gaussian distribution.
Finally, in three slots, user~$k$, $k=1,2$, can recover the messages
$(W_{k,1}, W_{k,2})$ sent in the
equivalent MIMO channel corresponding to the MAT alignment as well as
two fresh messages $(W_{k,3}, W_{k,4})$, 
from which the average DoF per user per channel use is
\begin{align}
  d_{\text{sym}} &= \frac{\dmimo + 2 {d}_{\text{p}k}}{3} = \frac{2 - \alpha + 2\alpha}{3} = \frac{2 + \alpha}{3}. 
\end{align}%
This concludes the achievability of the whole region given by~\eqref{eq:region} and Fig.~\ref{fig:DoF-Region}.

\begin{remark}
  By removing the private messages, one can send the common message in a
  higher rate~(corresponding to $d_\text{c}=1$ instead of
  $d_\text{c}=1-\alpha$) and thus shorten the
  communication~($1+2(1-\alpha)$ slots instead of $3$
  slots). This is the original idea reported in \cite{SubmittedISIT}
  that provides an achievable DoF of $\frac{2-\alpha}{3-2\alpha}$. Inspired by the gap
  between this DoF and the upper bound given by the converse
  \begin{align}
    \frac{2-\alpha}{3-2\alpha} \quad \text{versus} \quad
    \frac{2+\alpha}{3} = \frac{2-\alpha  {\color{red}\, +\, 2\alpha}}{3 -
    2\alpha  {\color{red}\,+\,2\alpha}}, 
  \end{align}%
  a natural question arose: \emph{Can we convey $2\alpha$ more symbols
  per user by extending the transmission by $2\alpha$ channel uses, i.e., in total over three channel
  uses?} It turned out that it is possible by exploiting the current
  CSI, according to Lemma~\ref{lemma:BC-CM}. 
\end{remark}

\begin{figure}
 \centering
\includegraphics[width=0.48\textwidth]{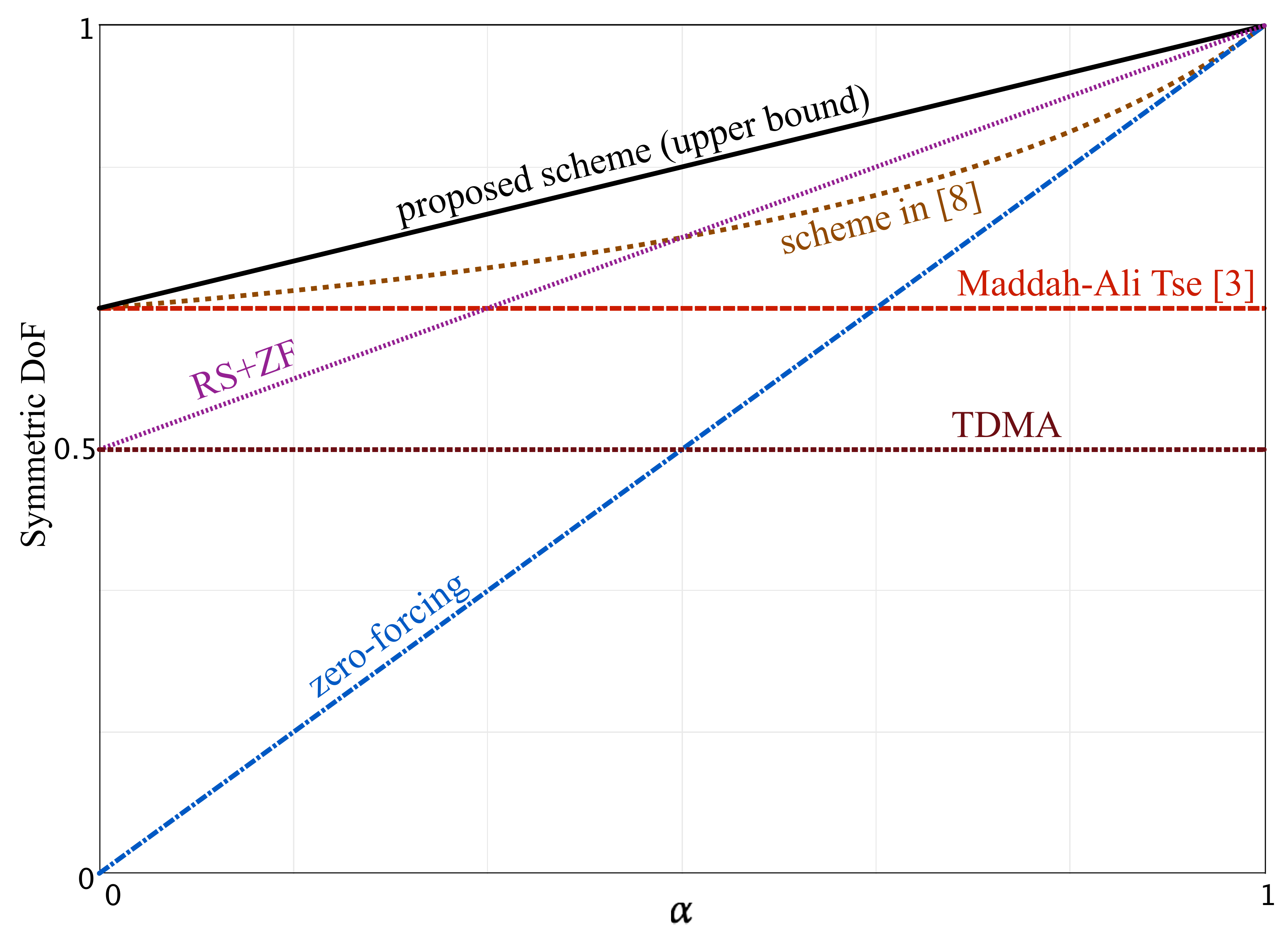}
\caption{Comparison of the achievable DoF between the proposed scheme and the zero-forcing and MAT alignment as a function of $\alpha$.}
\label{fig:DoF}
\end{figure}
 
In Fig.~\ref{fig:DoF}, we compare the achievable DoF of different
schemes.  The TDMA~(time sharing between single-user communications)
requires neither the current nor the delayed CSIT and achieves a
DoF of $\frac{1}{2}$. The ZF precoding only exploits the current CSIT
with a DoF of $\alpha$, while the MAT scheme only exploits the delayed
CSIT with a DoF of $\frac{2}{3}$. The scheme ``RS+ZF''~(Rate-Splitting and ZF precoding)
is from equally time sharing between the corner points $(1,\alpha)$ and $(\alpha,1)$. 
It only exploits the current CSIT with a DoF of
$\frac{1+\alpha}{2}$. 
Note that when $\alpha$ is close to $0$, the estimation of current CSIT is bad and
therefore useless. In this case, the optimal scheme is the 
MAT alignment. On the other hand, when $\alpha\ge1$, the estimation is good and
the interference at the receivers due to the imperfect estimation is
below the noise level and thus can be neglected as far as the DoF is
concerned. In this case, delayed CSIT is useless and even ZF with the
estimated current CSIT is
asymptotically optimal, achieving a DoF of $1$ per user.
Our result reveals that strictly larger DoF than
$\max\{\frac{2}{3}, \alpha\}$ can be obtained by exploiting both the
imperfect current CSIT and the perfect delayed CSIT in an intermediate
regime $\alpha\in(0,1)$. 

In the appendix, we provide the exact achievable rate region. Some
examples of the achievable sum rates with Rayleigh fading are shown in
Fig.~\ref{fig:ergodic-alpha}~and Fig.~\ref{fig:ergodic}.\footnote{Note
that the parameters are fixed according to the choices given in the appendix without
optimization.} In Fig.~\ref{fig:ergodic-alpha}, we plot the sum rate performance of our
sum-DoF optimal scheme for different values of $\alpha$. We observe that as the
quality of channel knowledge increases~($\alpha\to1$), the sum rate
improves significantly with the sharper slope promised by the DoF
result. Note that the performance with $\alpha=0$ nearly corresponds
to the sum rate achieved by MAT~(cf.~Fig.~\ref{fig:ergodic}). In
Fig.~\ref{fig:ergodic}, we compare our sum-DoF optimal scheme with 
different strategies: 
MAT, ZF, TDMA, as well as ``RS+ZF'' in terms of the ergodic sum rate for
$\alpha=0.5$.  For this quality of the current CSIT, ZF performs
substantially worse than the others, achieving the pre-log of one. With
the same value of DoF as ZF, the TDMA scheme performs much better than the ZF
scheme, since full transmit power can be used without causing
interference. Note that the current CSIT is exploited in the TDMA scheme
in such a way that the signal is beamformed in the direction of the estimated channel. The
sum rate with MAT, RS+ZF, and the proposed scheme increases with a slope
of $\frac{4}{3}$, $\frac{3}{2}$, and 
$\frac{5}{3}$, respectively, as expected from the DoF results.

\begin{figure}
 \centering
\includegraphics[width=0.48\textwidth]{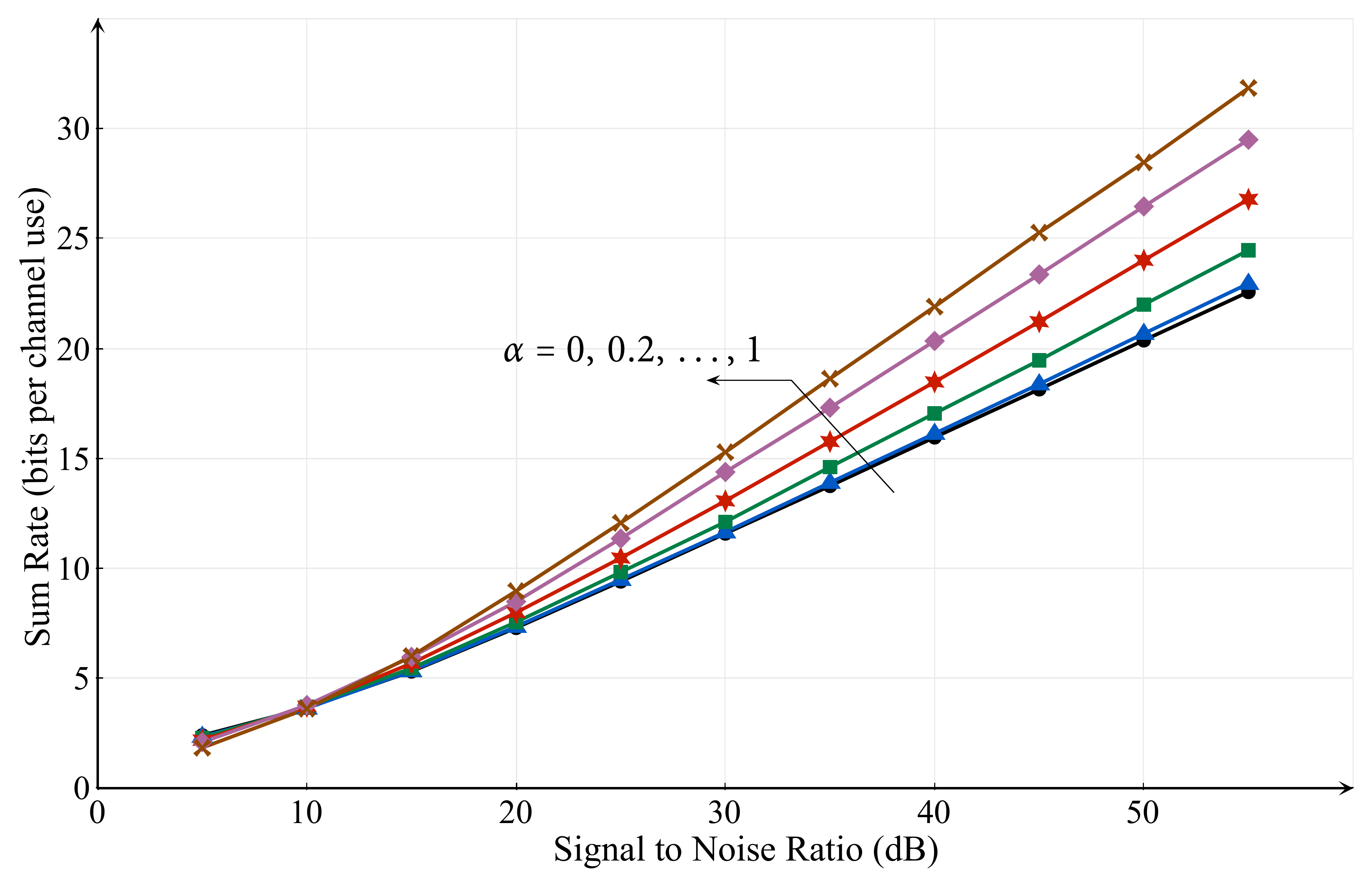}
\caption{The achievable ergodic sum-rate of the proposed scheme with
Rayleigh fading, for $\alpha=0, 0.2, \ldots, 1$.}
\label{fig:ergodic-alpha}
\end{figure}

\begin{figure}
 \centering
\includegraphics[width=0.48\textwidth]{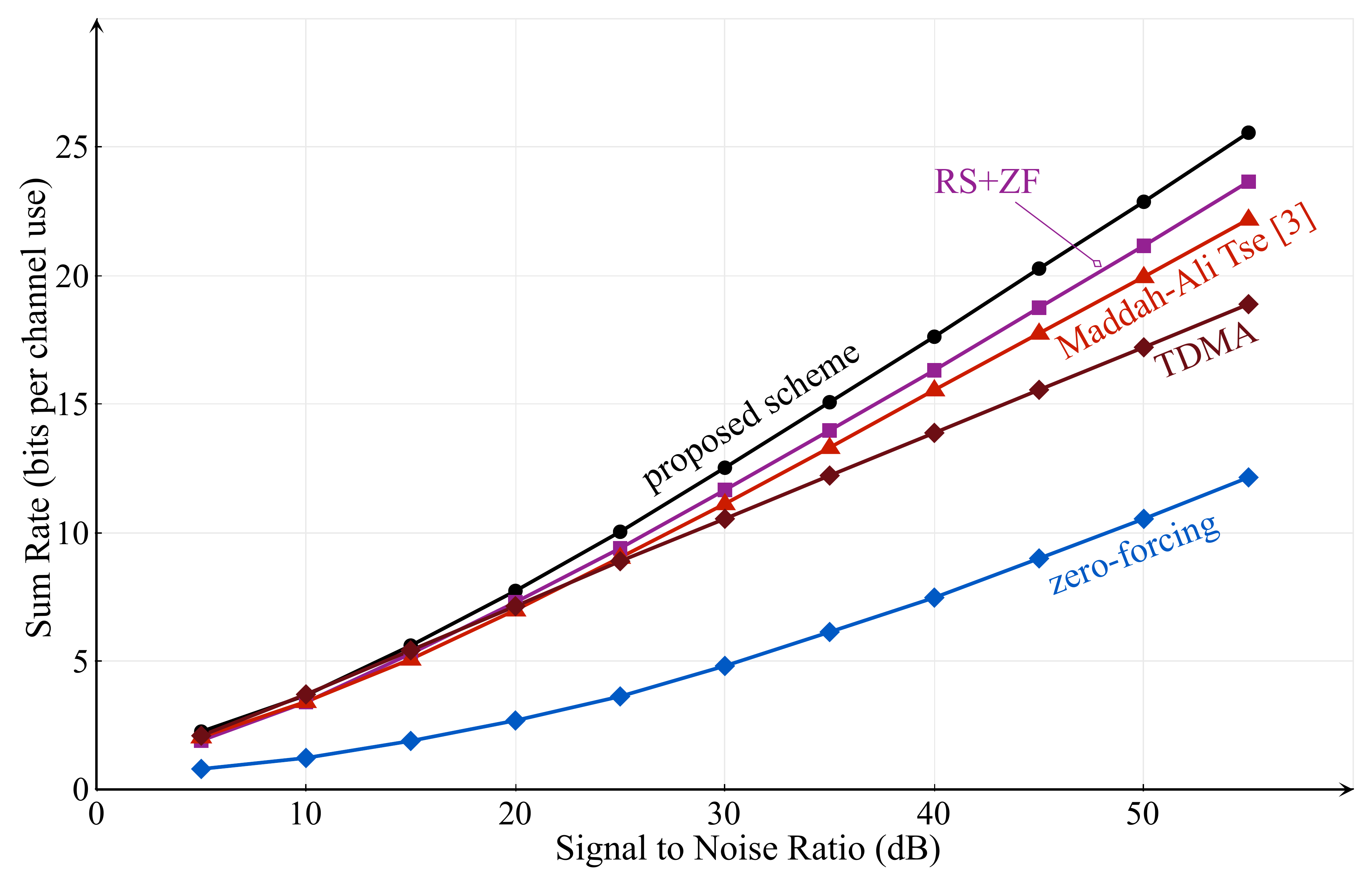}
\caption{The achievable ergodic sum-rate of the
proposed sum-DoF optimal scheme, rate-splitting scheme, TDMA, zero-forcing, and MAT alignment. We set $\alpha=0.5$.}
\label{fig:ergodic}
\end{figure}

\section{Discussions}
\label{sec:extension}

\subsection{DoF with common message}
\label{sec:DoF-common}

The main result of this paper can be extended trivially to the case with
common message. 
\begin{corollary}
  Let $(d_0, d_1, d_2)$ be the degrees of freedom related to the common
  message, private message for user~1, and private message for user~2,
  respectively. Then, the optimal DoF region 
  is  characterized by 
  \begin{subequations}
  \begin{align}
    d_0 + d_1 &\le 1, \label{eq:d0+d1}\\
    d_0 + d_2 &\le 1, \label{eq:d0+d2}\\
    2d_0 + d_1 + 2 d_2 &\le 2 + \alpha, \label{eq:2d0+d1+2d2}\\
    2d_0 + 2 d_1 + d_2 &\le 2 + \alpha. \label{eq:2d0+2d1+d2} 
  \end{align}%
  \label{eq:region-CM}
  \end{subequations}
\end{corollary}
\vspace{-0.5\baselineskip}
\begin{proof}
The converse follows the same lines as in the case without common
message, presented in Section~\ref{sec:converse}. To obtain
\eqref{eq:d0+d2} and \eqref{eq:2d0+d1+2d2}, we replace $W_2$ by
${W}'_2\defeq (W_0,W_2)$ and $R_2$ by ${R}'_2 \defeq R_0+R_2$
throughout Section~\ref{sec:converse} and carry out exactly the same
steps. Then, \eqref{eq:d0+d1} and \eqref{eq:2d0+2d1+d2} follow straightforwardly by
interchanging the roles of user~1 and user~2 as well as the symmetry
between the two users.  

Note that the region is a polyhedron and completely characterized by the vertices
in terms of $(d_0, d_1, d_2)$: 
\begin{itemize}
  \item extreme points: $(1, 0, 0)$, $(0, 1, 0)$, $(0, 0, 1)$,
  \item private points: $(0, 1, \alpha)$, $(0, \alpha, 1)$,
    $\left(0,
    \frac{2+\alpha}{3}, \frac{2+\alpha}{3}\right)$, and  
  \item mixed point: $(1-\alpha, \alpha, \alpha)$
\end{itemize}
which are all achievable with the proposed scheme. Thus, the entire
region is achievable by time sharing between the vertices. 
\end{proof}

\subsection{Imperfect delayed CSI: Limited feedback}

In most practical scenarios, delayed CSIT is obtained through feedback
channel and the current state is then predicted based on the delayed
CSIT. Due to various reasons, perfect delayed CSIT may not be available. 
For instance, the limited feedback rate may incur a
distortion on the channel coefficients. In the following, we take a look
at the impact of the imperfect delayed CSIT on the achievable DoF of the proposed scheme. 

First, let us assume that the channel state $\Sm_{t-1}$ is quantized before being sent back to
the transmitter~(and to the other receiver). The quantization model is  
\begin{align}
  \Sm_{t-1} &= \bar{\Sm}_{t-1} + \breve{\Sm}_{t-1}
\end{align}%
where each entry of the quantization noise $\breve{\Sm}_{t-1}$ has the same
variance $\sigma^2_{\text{FB}}$. 
We introduce a parameter $\beta$ to characterize the precision of the
quantization. As the definition of $\alpha$, we define $\beta$ as the
power exponent of the quantization noise\footnote{From the rate-distortion function, it is not difficult to relate $\beta$ to the resource required for the CSI feedback, i.e., the feedback DoF.}, i.e., 
\begin{align}
  \beta &\defeq \min\left\{ -\frac{\log \sigma^2_{\text{FB}}}{\log P},\, 1
  \right\}.
\end{align}%

Due to the lack of perfect delayed CSIT, instead of using $\Sm^{t-1}$ to
predict $\Sm_t$ for the precoding and using $\Sm_{t-1}$ to perform the
MAT alignment, the transmitter now predicts the quantized state $\bar{\Sm}_t$ with
the past quantized state $\bar{\Sm}^{t-1}$ and uses $\bar{\Sm}_{t-1}$ for the alignment.
Therefore, although the actual interference seen by the receivers is
$(\hv^\H \vv, \gv^\H \uv)$, the transmitter only has access to a noisy
version of it $\etav = (\bar{\hv}^\H \vv, \bar{\gv}^\H \uv)$. Receiver~1
has the following equations
\begin{align}
  y &= \hv^\H \uv + \hv^\H \vv + \noisey = \hv^\H \uv + \eta_1 +
  (\hv-\bar{\hv})^\H \vv + \noisey,\IEEEeqnarraynumspace  \label{eq:new_y1}\\
  \hat{\eta}_1 &= \eta_1 - {\Delta_1}, \\
  \hat{\eta}_2 &= \eta_2 - {\Delta_2} = \bar{\gv}^\H \uv - {\Delta_2}. 
\end{align}%
The power of $\etav$ is $\bar{\hv}^\H\Qv\bar{\hv} + \bar{\gv}^\H \Qu
\bar{\gv}$ that depends on the ``precision'' of the prediction from
$\bar{\Sm}^{t-1}$ to $\bar{\Sm}_{t}$. It can be shown\footnote{Without
going into the details, we can see that the following Markov chain holds
$\bar{\Sm}^{t-1}\leftrightarrow{\Sm}^{t-1}\leftrightarrow{\Sm}_{t}\leftrightarrow\bar{\Sm}_t$.
The prediction error from $\bar{\Sm}^{t-1}$ to $\bar{\Sm}_{t}$ is now the 
aggregation of two effects: the channel variation, characterized by
$P^{-\alpha}$, and the quantization error due to limited feedback rate,
characterized by $P^{-\beta}$. Hence, we have the power exponent of the
aggregated error $\alpha' = \min\{\alpha, \beta\}$. } that the power
exponent of this prediction error is $\alpha'\defeq\min\{\alpha, \beta\}$ where
$\alpha$ is the power exponent of the prediction error when perfect
delayed CSIT is present, i.e., predicting ${\Sm}_{t}$ from ${\Sm}^{t-1}$.  
Therefore, the achievable DoF of the proposed scheme would be
$\frac{2+\alpha'}{3}$ without taking into account the ``residual
interference'' $(\hv-\bar{\hv})^\H \vv$ in \eqref{eq:new_y1}. In fact, this
interference costs a DoF loss of $1-\beta$ over three slots, yielding the new DoF per
user 
\begin{align}
  d(\alpha,\beta) &= \frac{2+\alpha'-(1-\beta)}{3} \\
  &= \frac{1+\min\left\{ \alpha, \beta \right\} +
  \beta}{3}, \quad \alpha,\beta\in[0,1]. 
\end{align}%

\begin{figure}[t]
 \centering
\includegraphics[width=0.48\textwidth]{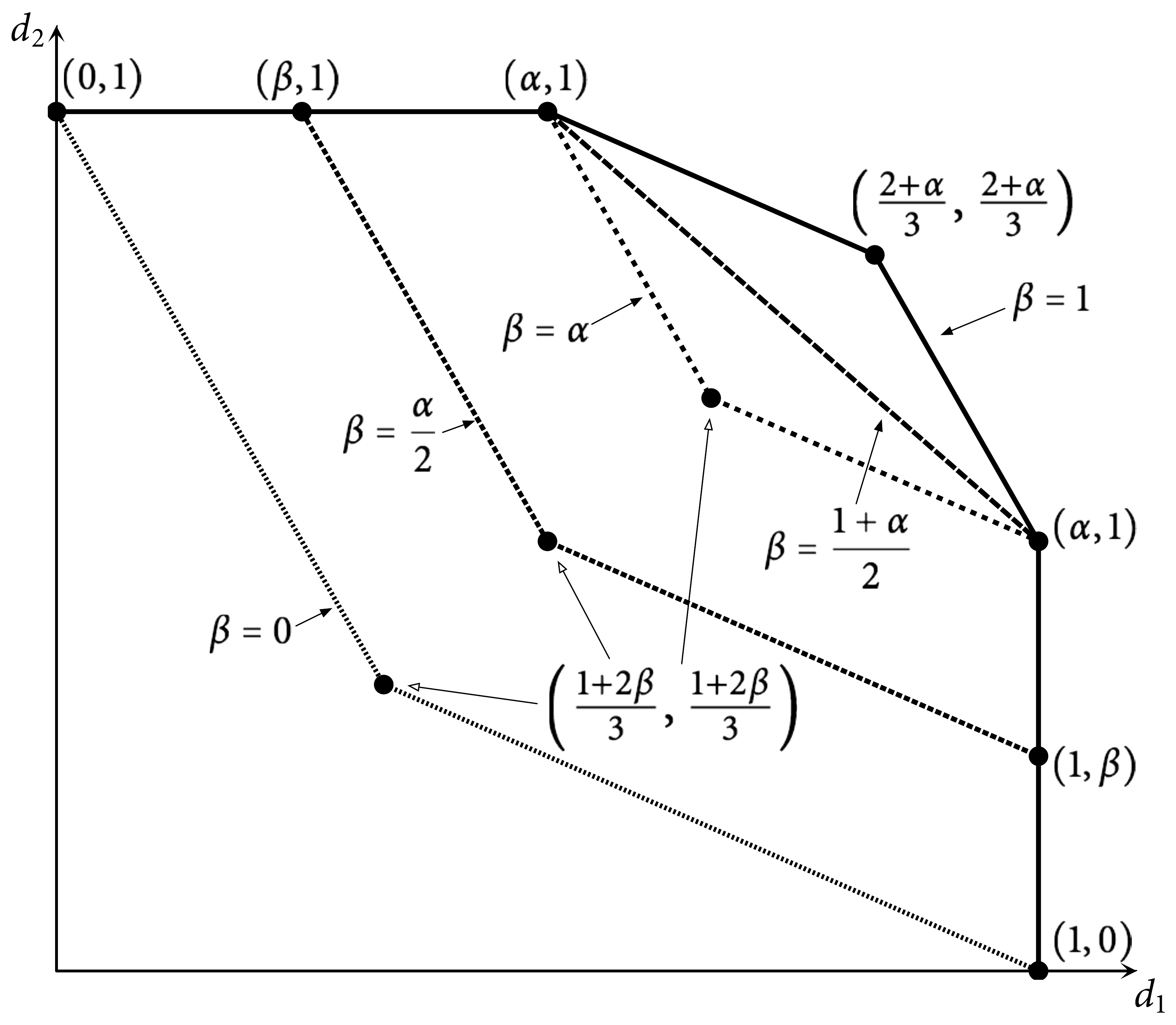}
\caption{Impact of imperfect delayed CSIT on the achievable DoF with the
proposed scheme. We fix
$\alpha=0.5$ and vary $\beta$ from $1$ to $0$. }
\label{fig:DoF_beta}
\end{figure}

As in the case with perfect delayed CSIT, the DoF pairs $(1, \alpha')$
and $(\alpha', 1)$ are achievable without the MAT alignment. An example
of the DoF region is shown in Fig.~\ref{fig:DoF_beta}, where we fix the
value $\alpha$ and vary $\beta$ from $1$ to $0$. As shown in the figure,
when $\beta = 1$, the DoF region is unchanged. When $\beta$ is reduced
to $\frac{1+\alpha}{2}$, the symmetric DoF point can be achieved by time
sharing between the two corner points $(1,\alpha)$ and $(\alpha,1)$.
Delayed CSIT is not beneficial any more with our scheme. As $\beta$ continues to
diminish to $\alpha$, the symmetric DoF keeps dropping while the
corner points remain still. At this point, using MAT alignment creates
more interference than resolving it. When $\beta$ goes below $\alpha$,
it becomes the dominating source of interference. The corner points
become $(1, \beta)$ and $(\beta, 1)$. The above analysis reveals that
even imperfect delayed CSIT can be beneficial with our scheme, as long as the
feedback accuracy $\beta$ is larger than $\frac{1+\alpha}{2}$.
However, it is unclear whether this naive extension to the
imperfect delayed CSIT case is optimal. Finding optimal schemes with
imperfect delayed CSIT remains an open problem and is out of the
scope of this paper.

\subsection{Bandwidth-limited Doppler process}

The main result on the achievable DoF has been presented in terms of an
artificial parameter $\alpha$, denoting the speed of decay of the
estimation error $\sigma^2\sim P^{-\alpha}$ in the current CSIT. In this
section, we provide an example showing the practical interpretation of
this parameter. Focusing on receiver 1 due to symmetry, we describe the
fading process, channel estimation, and feedback scheme as follows: 
\begin{itemize}
  \item The channel fading $\hv_t$ follows a Doppler
    process with power spectral density $S_h(w)$. The channel coefficients
    are strictly band-limited to $[-F,F]$ with $F=\frac{v f_c 
    T_f}{c} < \frac{1}{2}$ where $v, f_c, T_f$, and $c$ denote the mobile speed
    in $\text{m}/\text{sec}$, the carrier frequency in $\text{Hz}$, the
    slot duration in $\text{sec}$, the light speed in
    $\text{m}/\text{sec}$, respectively.
  \item The channel estimation is done at the receivers side with
    pilot-based downlink training. At slot $t$, receiver 1
    estimates $\hv_{t}$ based on a sequence of noisy observations
    $\bigl\{\sv_{\tau} = \sqrt{P} \hv_{\tau} + \nuv_{\tau} \bigr\}$ up to $t$, where 
    $\nuv_t\sim\CN(0, \Id)$ is the AWGN. The estimate is denoted by $\bar{\hv}_t$ with  
    \begin{align}
      \hv_t &= \bar{\hv}_t + \breve{\hv}_t. \label{eq:estimation}
    \end{align}%
    Under this model, the estimation error vanishes as
    $\E\bigl(\norm{\breve{\hv}_t}^2\bigr) \sim P^{-1}$.
  \item At the end of slot $t$, the noisy observation 
    $\sv_t$ is sent to the transmitter and receiver 2 over a noise-free channel. 
    At slot $t+1$, based on the noisy observation $\{\sv_{\tau}\}$ up to $t$, 
     the transmitter and receiver 2 acquire the prediction ${\hat{\hv}}_{t+1}$ of $\hv_{t+1}$ and estimation $\tilde{\hv}_t$ of $\hv_{t}$. The corresponding prediction model is  
    \begin{align}
      {\hv}_{t} &= {\hat{\hv}}_{t} +
      {\tilde{\hv}}_{t}. \label{eq:prediction}
    \end{align}%
    From \cite[Lemma 1]{caire2010multiuser}, we have 
    $\E\bigl(\norm{{\tilde{\hv}}_t}^2\bigr) \sim P^{-(1-2F)}$. 
\end{itemize}

In this channel with imperfect delayed CSIT, we can still apply the
proposed scheme and analysis in exactly the same way as in the previous
section with $\alpha=1-2F$ and $\beta=1$.  

\subsection{Non-ergodic fading~(delay-limited communications)}

The DoF results have been derived based on the ergodic rates. For
non-ergodic fading processes, the DoF can be redefined in the same
manner as the definition of multiplexing gain in \cite{Zheng_Tse}. This
approach has been reported in \cite{SubmittedISIT}. Following the
footsteps in
\cite{SubmittedISIT}, it can be shown that the non-ergodic DoF coincides with the ergodic DoF.

\section{Conclusions}
\label{sec:conclusions}

A scheme achieving the optimal degrees of freedom region in a two-user
MISO broadcast channel has been presented. The approach optimally
exploits the combination of  delayed channel feedback together with
imperfect current CSIT. In practical scenarios, the current CSIT may be
obtained from a prediction based on the delayed CSIT samples. When the
quality of current CSIT is poor, the proposed scheme coincides with the
previously reported MAT space-time alignment, whereas as the current CSIT
prediction quality becomes ideal, the scheme relies on standard linear
precoding. In between these extremal regimes, the proposed strategy
advocates interference quantization followed by feedback.
Generalizations of the proposed study to the MIMO case, multi-user case,
and imperfect delayed CSIT case remain challenging yet interesting open
problems. 

\appendix

\newcommand{\Xim}{\pmb{\Xi}}
\newcommand{\Gammam}{\mat{\Gamma}}
\newcommand{\Xmc}{\Xc_{\text{mc}}}
\newcommand{\Tu}{\Thetam}
\newcommand{\Tpi}{\Xim}
\newcommand{\Tpj}{\Gammam}
\newcommand{\Tv}{\Phim}
\newcommand{\Tmc}{\Omegam}
\newcommand{\Lu}{\Lambdam_u}
\newcommand{\Lpi}{\Lambdam_{\text{p}1}}
\newcommand{\Lpj}{\Lambdam_{\text{p}2}}
\newcommand{\Lpk}{\Lambdam_{\text{p}k}}
\newcommand{\Lv}{\Lambdam_v}
\newcommand{\Lzero}{\Lambdam_\text{c}}
\newcommand{\Lmc}{\Lambdam_{\text{mc}}}
\newcommand{\Wmc}{W_{\text{mc}}}

\subsection{Proof of Lemma \ref{lemma:logdet}} \label{app:proof-lemma-logdet}
  First, we show \eqref{eq:tmp89} as follows.
  \begin{align}
    \MoveEqLeft \E_{S_i \vert \hat{S}_i} \left( \log(1+\hv_i^\H \Km
    \hv_i) \right)  \\
    &\le \E_{S_i \vert \hat{S}_i} \left( \log(1+ \lambda_1 \Norm{\hv_i}) \right)  \\
    &\le \log(1+ \lambda_1 \Norm{\hat{\hv}_i} + m\sigma^2\lambda_1 )
    \label{eq:tmpconc} \\
    &= \log(1+ \lambda_1 \Norm{\hat{\hv}_i}) + \log\biggl(1 +
    \frac{m\sigma^2\lambda_1}{1+\Norm{\hat{\hv}_i} \lambda_1} \biggr)\\    
    &\le \log(1+ \lambda_1 \Norm{\hat{\hv}_i}) + \log\biggl(1 +
    \frac{m\sigma^2}{\Norm{\hat{\hv}_i}} \biggr)
  \end{align}%
  where \eqref{eq:tmpconc} is from the concavity of the log function. 
  
  Then, to derive \eqref{eq:tmp99}, let us define $\hat{\psiv}\defeq \Vm^\H
  \hat{\gv}_i$ and $\tilde{\psiv} \defeq \Vm^\H
  \tilde{\gv}_i$ with $\Vm$ being the unitary matrix containing the
  eigenvectors of $\Km$, i.e., $\Km = \Vm \diag[\lambda_1,\ldots, \lambda_m]
  \Vm^\H$. From the isotropic assumption, $\tilde{\psiv}$ has the same
  distribution as $\tilde{\gv}_i$ and is also isotropic. 
  Since the distribution of the vector $\tilde{\psiv}$ is invariant
  under unitary transformations, it follows that the distribution of each
  scalar $\tilde{\psi}_l$ in $\tilde{\psiv}$ is invariant under complex
  scalar rotations. Thus, $\tilde{\psi}_l$, $l=1,\ldots,m$, can be represented by
  $A_l e^{j\theta_l}$ where $A_l\defeq \abs{\tilde{\psi}_l}$ is
  independent of $\theta_l$ that is uniformly distributed in $[0, 2\pi)$. We
  need the following lemma for the proof. 
  \begin{lemma}\label{lemma:BA}
    Let $\theta$ be a random variable uniformly distributed in
    $[0, 2\pi)$. Then, we have 
    \begin{align}
      \E_\theta\bigl(\log\bigl(\abs{B + A e^{j\theta}}^2\bigr)\bigr)
       &= \log \bigl(\max\bigl\{ \abs{A}^2, \abs{B}^2 \bigr\}\bigr). 
    \end{align}%
  \end{lemma}
\vspace{0.3\baselineskip}
  \begin{proof}
    Without loss of generality, we assume that both $A$ and $B$ have
    non-negative real values, since $\theta$ is uniformly distributed in
    $[0, 2\pi)$. The expectation $\E_\theta\bigl(\log\bigl(\abs{B + A
    e^{j\theta}}^2\bigr)\bigr)$ can be directly calculated as follows:
    \begin{align}
       \MoveEqLeft \E_\theta\bigl(\log(\abs{B + A
       e^{j\theta}}^2)\bigr) \\ &= 
       \E_\theta\bigl(\log(A^2 + B^2 + 2 A B \cos(\theta))\bigr) \\
       &= \frac{1}{2\pi} \int_0^{2\pi} \log\bigl(A^2 + B^2 + 2 A B \cos(\theta) \bigr) \text{d}\theta\\
       &= \log\frac{A^2 + B^2 + \sqrt{(A^2+B^2)^2 -
       (2AB)^2}}{2} \label{eq:int}\\
       &= \log \bigl(\max\bigl\{ \abs{A}^2, \abs{B}^2 \bigr\}\bigr) 
    \end{align}%
    where \eqref{eq:int} is from the identity $$ \int_{0}^1 \log(a+b
  \cos(2\pi t)) \text{d} t = \log\frac{a+\sqrt{a^2 - b^2}}{2},\
  \forall\, a\ge b > 0.$$ 
   \end{proof}
  Now, we can finish the proof of \eqref{eq:tmp99} as follows:
  \begin{align}
    \MoveEqLeft{\E_{S_i \vert \hat{S}_i} \bigl( \log(1+\gv_i^\H \Km
    \gv_i) \bigr)} \\
    &= \E_{\tilde{\psiv}\vert\hat{S}_i}\biggl( \log\biggl(1 + \sum_{j=1}^m \lambda_j \Abs{\hat{\psi}_j
    + \tilde{\psi}_j} \biggr) \biggr)  
    \label{eq:tmpl1}
    \\
    &\ge \E_{\tilde{\psi}_1\vert\hat{S}_i}\!\bigl( \log(\lambda_1\Abs{\hat{\psi}_1
    + \tilde{\psi}_1}) \bigr)^+ 
    \label{eq:tmpl2}
    \\
    &\ge \Bigl( \E_{\tilde{\psi}_1\vert\hat{S}_i}\bigl( \log(\lambda_1\Abs{\hat{\psi}_1
    + \tilde{\psi}_1}) \bigr) \Bigr)^+ \label{eq:tmpl3} \\
    &\ge \Bigl( \E_{{\tilde{\psi}_1}\vert\hat{S}_i}\bigl(
    \log(\lambda_1\Abs{\tilde{\psi}_1}) \bigr) \Bigr)^+ \label{eq:tmpl4} \\
    &= \left( \log(2^{\gamma}\sigma^2\lambda_1) \right)^+ \label{eq:tmpl5} \\
    &\ge \log(1+2^{\gamma}\sigma^2\lambda_1) - 1 \label{eq:tmpl6} 
\end{align}
where in \eqref{eq:tmpl2}, $(x)^+$ means $\max\left\{ x, 0 \right\}$; 
\eqref{eq:tmpl3} is from the fact that moving the maximization
outside of the expectation does not increase the value; \eqref{eq:tmpl4}
is obtained by using the fact that $\tilde{\psi}_1$ is invariant under
complex scalar rotations and by applying Lemma~\ref{lemma:BA}~(averaging over the phase of
$\tilde{\psi}_1$); 
in \eqref{eq:tmpl5}, we define $\gamma \defeq \E_{\tilde{\psi}_1 \vert
\hat{S}_i} \Bigl( \log \frac{\Abs{\tilde{\psi}_1}}{\sigma^2} \Bigr) =
\E_{\tilde{S}_{i} \vert
\hat{S}_i} \Bigl( \log \frac{\Abs{\tilde{g}_{i,1}}}{\sigma^2}
\Bigr)$ with $\gamma>-\infty$ according to
Assumption~\ref{assumption:fading};
in \eqref{eq:tmpl6}, we apply the inequality $\bigl(\log(x)\bigr)^+ \ge \log(1+x) - 1$.

\subsection{Proof of Lemma~\ref{lemma:BC-CM}}
\label{app:BC-CM}

We describe the coding scheme in Lemma~\ref{lemma:BC-CM} as follows. 
\begin{itemize}
  \item Channel codebooks $\Xc_\text{c}, \Xc_{\text{p}1}, \Xc_{\text{p}2}$ of length $n$ and
    sizes $2^{n R_\text{c}}$, $2^{n R_{\text{p}1}}$, and $2^{n R_{\text{p}2}}$,
    respectively. Entries of these codebooks are generated 
    i.i.d. according to $\CN[0, \Lzero]$, $\CN[0, \Lpi]$, and $\CN[0,
    \Lpj]$, respectively, with $\Lzero,\Lpi,\Lpj\succeq0$ being $m\times m$ matrices that can be assumed to be diagonal without loss of generality. 
  \item Time-varying linear precoders that only depend on the estimate of
    the current state: 
    $$\Tpi_t, \Tpj_t, \Tmc_t: \hat{\Sc}_t \longmapsto \CC^{m\times m}.$$
  \item Coding: The commom message denoted by $W_\text{c}$ is coded
    in $\{\tilde{\xv}_{\text{c},t}\}_{t=1}^{n}\in\Xc_\text{c}$, precoded, 
    and then multicast to both users.
    Meanwhile, two private messages $W_{\text{p}1}$ and $W_{\text{p}2}$ for user~1
    and user~2, respectively, are coded in $\{\tilde{\uv}_{\text{p},t}
    \}_{t=1}^{n}\in\Xc_{\text{p}1}$ and $\{\tilde{\vv}_{\text{p},t}
    \}_{t=1}^{n}\in\Xc_{\text{p}2}$, respectively, precoded, and sent. The
    transmitted signal is 
    \begin{align}
      \xv_{t}  &= \Omegam_t \tilde{\xv}_{\text{c},t}+
      \Tpi_t \tilde{\uv}_{\text{p},t} + \Tpj_t \tilde{\vv}_{\text{p},t},
      \quad t=1,\ldots,n.
    \end{align}
\end{itemize}

Then, we can get the following achievable rate region. 
\begin{proposition}\label{prop:BC-CM}
  The achievable rate region of the two-user MISO broadcast channel with
  common message is the union of the rate triples $(R_\text{c}, R_{\text{p}1},
  R_{\text{p}2})$ with
  \begin{align}
    R_\text{c} &\defeq \min\Biggl\{ \E\Biggl( \log\Biggl( 1+ \frac{\hv^\H \Qm_\text{c} \hv}{ 1+
    \hv^\H (\Qm_{\text{p}1}+\Qm_{\text{p}2}) \hv } \Biggr)\Biggr) ,\\ 
    &\qquad \E[ \log\left( 1+
    \frac{\gv^\H \Qm_\text{c} \gv}{ 1+ \gv^\H (\Qm_{\text{p}1}+\Qm_{\text{p}2}) \gv }
    \right)] \Biggr\}, \label{eq:BC-CM-a}\\
    R_{\text{p}1} &\defeq \E\Biggl( \log\Biggl( 1+ \frac{\hv^\H \Qm_{\text{p}1}
    \hv}{ 1 + \hv^\H \Qm_{\text{p}2} \hv} \Biggr) \Biggr),
    \label{eq:BC-CM-b}\\ 
    R_{\text{p}2} &\defeq \E\Biggl( \log\Biggl( 1+ \frac{\gv^\H \Qm_{\text{p}2} \gv}{ 1 +
    \gv^\H \Qm_{\text{p}1} \gv} \Biggr)\Biggr), \label{eq:BC-CM-c}
  \end{align}%
  over all policies 
  \begin{align}
    Q(\hat{\Sm}) &\defeq \left\{\Qm_\text{c}, \Qm_{\text{p}1}, \Qm_{\text{p}2} \succeq 0:\
    \trace[\Qm_\text{c}+\Qm_{\text{p}1}+\Qm_{\text{p}2}] \le P \right\}
  \end{align}%
  that only depend on the estimate of the channels $\hat{\Sm}$.
\end{proposition}

\begin{proof}
  The proof is straightforward. First, the common message is decoded by treating the private signals as
  noise. Then, after removing the decoded common signal, the private
  message is obtained by treating the interference as noises. The covariance matrices are such that 
  $\Qm_\text{c} = \Tmc \Lzero \Tmc^\H$, $\Qm_{\text{p}1} = \Tpi \Lpi \Tpi^\H$,
  $\Qm_{\text{p}2} = \Tpj \Lpj \Tpj^\H$.  Further details are omitted.  
\end{proof}

Setting $\Qm_\text{c}\sim P \Id$, $\Qm_{\text{p}1} \sim P^\alpha \Psimgp$, and
$\Qm_{\text{p}2} \sim P^\alpha \Psimhp$, Lemma~\ref{lemma:BC-CM} follows
immediately.

\subsection{Achievable rate region of the sum-DoF optimal scheme}
\label{app:exactRate}

Let us recall that the proposed scheme consists of two phases. 
In the following, we let $n_1$ and $n_2$ denote the length of
Phase~1 and Phase~2, in channel uses, respectively.  
The main ingredients in Phase~1 are:
\begin{itemize}
  \item Codebook generation: 
    \begin{itemize}
      \item Channel codebooks $\Xc_{\tilde{\uv}}$ of length $n_1$ and
        size $2^{n_1 \Rmimoi}$, $\Xc_{\tilde{\vv}}$ of length $n_1$ and
        size $2^{n_1 \Rmimoj}$. Entries of $\Xc_{\tilde{\uv}}$ and
        $\Xc_{\tilde{\vv}}$ are generated i.i.d.  according to $\CN[0,
        \Lu]$ and $\CN[0, \Lv]$, respectively. $\Lu, \Lv\succeq0$ are
        $m\times m$ diagonal matrices. 
      \item Source codebooks $\Cc_k$ of length $n_1$ and size
        $2^{n_1 R_{\eta_k}}$, $k=1,2$. Entries of $\Cc_1$ and
        $\Cc_2$ are generated i.i.d. according to $\CN\bigl(0,
        1-\tilde{D}_k\bigr)$, $\tilde{D}_k\le1$, $k=1,2$. 
    \end{itemize}
  \item Time-varying linear precoders that only depend on the estimate of
    the current state: 
    \begin{align}
    \Tu_t, \Tv_t: \hat{\Sc}_t \longmapsto \CC^{m\times m}.
    \end{align}%
  \item Coding in Phase~1: The codewords $\left\{ \tilde{\uv}_t
    \right\}_{t=1}^{n_1}$ and $\left\{ \tilde{\vv}_t
    \right\}_{t=1}^{n_1}$ are
    selected from $\Xc_{\tilde{\uv}}$ and $\Xc_{\tilde{\vv}}$, according
    to $W_{\text{mimo},1}$ and $W_{\text{mimo},2}$, 
    respectively. The transmitted signal is 
    \begin{align}
      \xv_t  
      &= \Tu_t \tilde{\uv}_t + \Tv_t \tilde{\vv}_t, \quad
      t=1,\ldots,n_1. 
    \end{align}%
  \item Quantization of the interferences $\eta_1$ and $\eta_2$: At the end of
    Phase~1, the transmitter knows $\{(\eta_{1,t},
    \eta_{2,t})\}_{t=1}^{n_1}$
    with $\eta_{1,t} \defeq \hv_t ^\H\vv_t \sim \CN(0,\sigma_{\eta_{1,t}}^2)$ and
    $\eta_{2,t} \defeq \gv_t^\H \uv_t \sim \CN(0,\sigma_{\eta_{2,t}}^2)$, for a
    given channel realization $\{\hv_t, \gv_t\}_{t=1}^{n_1}$. The codebook
    $\Cc_k$, $k=1,2$, is used to quantize the \emph{normalized}
    source $\left\{\frac{\eta_{k,t}}{\sigma_{\eta_{k,t}}}
    \right\}_{t=1}^{n_1}$ that is i.i.d. $\CN[0,1]$. The quantized
    outputs are represented in $n_1(R_{\eta_1} + R_{\eta_2})$ bits. 
\end{itemize}
In Phase~2, exactly the same codebooks and precoders as in
Appendix~\ref{app:BC-CM} are used, except that the length of the
codewords is $n_2$ instead of $n$.  The quantized interferences,
represented in $n_1(R_{\eta_1} + R_{\eta_2})$ bits and denoted by
$W_\text{c}$,
is coded in
$\{\tilde{\xv}_{\text{c},t}\}_{t=n_1+1}^{n_1+n_2}\in\Xc_\text{c}$,
precoded, and then multicast to both users.  Meanwhile, two private
messages $W_{\text{p}1}$ and $W_{\text{p}2}$ for user~1 and 2 are coded in
$\{\tilde{\uv}_{\text{p},t} \}_{t=n_1+1}^{n_1+n_2}\in\Xc_{\text{p}1}$ and
$\{\tilde{\vv}_{\text{p},t} \}_{t=n_1+1}^{n_1+n_2}\in\Xc_{\text{p}2}$, respectively,
precoded, and sent. The transmitted signal is 
\begin{align}
  \xv_{t}  &= \Omegam_t \tilde{\xv}_{\text{c},t}+ \Tpi_t
  \tilde{\uv}_{\text{p},t} + \Tpj_t \tilde{\vv}_{\text{p},t}, \quad
  t=n_1+1,\ldots,n_1+n_2.
\end{align}

For user $k$ to recover its original messages $(W_{\text{mimo},k},
W_{\text{p}k})$ correctly\footnote{Note that the
assumption on the ergodicity and the Markov chain~\eqref{eq:Markov}
makes the single-letter representation of the rates possible.}, when $n_1,n_2\to\infty$, it is enough to 
\begin{itemize}
  \item recover the message $(W_\text{c}, W_{\text{p}k})$, which is possible
    if
    \begin{align}
      n_1(R_{\eta_1}+R_{\eta_2}) \le n_2 R_\text{c},  \label{eq:n1n2} 
    \end{align}%
    and if the triple $(R_\text{c}, R_{\text{p}1}, R_{\text{p}2})$ lies in the region defined in Proposition~\ref{prop:BC-CM}; 
  \item reconstruct $\left\{ \hat{\eta}_{k,t} \right\}_{t=1}^{n_1}$,
    $k=1,2$, with 
    \begin{align}
      \eta_{k,t} &= \hat{\eta}_{k,t} + \Delta_{k,t}, \quad
      \Delta_{k,t} \sim \CN\bigl(0, \sigma^2_{\eta_{k,t}}
      \tilde{D}_k\bigr), 
    \end{align}%
    which is possible if 
    \begin{align}
      R_{\eta_k} > \log\left(\frac{1}{\tilde{D}_k}\right), \quad k=1,2; 
    \end{align}%
  \item then decode the message $W_{\text{mimo},k}$, 
  which is possible if
    \begin{align}
      \Rmimoi &< I(\tilde{U}; Y, \hat{\eta}_1, \hat{\eta}_2 \cond S,
      \hat{S}), \label{eq:Rmimo1}\\
      \Rmimoj &< I(\tilde{V}; Z, \hat{\eta}_1, \hat{\eta}_2 \cond S,
      \hat{S}). \label{eq:Rmimo2}
    \end{align}%
\end{itemize}
Putting all pieces together, we obtain the rate region of the proposed
scheme in the following.  
\begin{proposition}
  Let $(R_\text{c}, R_{\text{p}1}, R_{\text{p}2})$ be defined as
  in Proposition~\ref{prop:BC-CM} and let us define the compression rate~$R_{\eta_k}$ and MIMO rate as  
  \begin{align}
    R_{\eta_k} &\defeq   \log \frac{1}{\tilde{D}_k}, \quad k=1,2, \label{eq:Retav}\\
    \Rmimoi &\defeq \E\bigl( \log \det \left( \Id + \Dm_1 \Sm \Qu \Sm^\H
    \right) \bigr), \label{eq:Rmimoi}\\
    \Rmimoj &\defeq \E\bigl( \log \det \left( \Id + \Dm_2 \Sm \Qv \Sm^\H
    \right)\bigr), \label{eq:Rmimoj}\\
    \shortintertext{with}
    \Dm_1 &\defeq  \diag[\frac{1}{1+\hv^\H \Qv \hv\,\tilde{D}_1},\
    \frac{1-\tilde{D}_2}{\gv^\H \Qu \gv\, \tilde{D}_2}], \\
    \Dm_2 &\defeq  \diag[\frac{1-\tilde{D}_1}{\hv^\H \Qv \hv\, \tilde{D}_1},\ \frac{1}{1+\gv^\H \Qu \gv\,\tilde{D}_2}].
  \end{align}%
  Then, the achievable rate region of the proposed scheme is the union of the rate pairs
  $(R_1, R_2)$ with
  \begin{align}
    R_k &= \frac{ R_\text{c} \,{\Rmimok} + (R_{\eta,1} + R_{\eta,2})
    R_{\text{p}k} }{R_\text{c} +  R_{\eta,1} + R_{\eta,2}}, \quad k=1,2,  \label{eq:R_scheme2}
  \end{align}%
  over all policies $D(\hat{\Sm}) \defeq \{\tilde{D}_1, \tilde{D}_2:
  \ 0\le\tilde{D}_k\le1\}$ and 
  \begin{align}
    Q'(\hat{\Sm}) &\defeq \bigl\{ \Qu, \Qv, \Qm_\text{c}, \Qm_{\text{p}1},
    \Qm_{\text{p}2} \succeq 0:\\
    &\qquad \trace(\Qu+\Qv) \le P,\
    \trace(\Qm_\text{c}+\Qm_{\text{p}1}+\Qm_{\text{p}2}) \le P \bigr \}
  \end{align}%
  that only depend on the estimate of the channels. 
\end{proposition}
\begin{proof}
The average achievable rate for user $k$ is 
\begin{align}
  R_k &= \frac{n_1\Rmimok + n_2 R_{\text{p}k}}{ n_1 + n_2} \\
  &= \frac{\displaystyle \Rmimok + \frac{n_2}{n_1}  R_{\text{p}k}}{ \displaystyle 1 + \frac{n_2}{n_1}} \\
  &= \frac{ R_\text{c} \,{\Rmimok} + (R_{\eta,1} + R_{\eta,2})
  R_{\text{p}k} }{R_\text{c} +  R_{\eta,1} + R_{\eta,2}} 
\end{align}%
where the last equality holds by choosing $n_1$ and $n_2$ that equalize \eqref{eq:n1n2}. 
To see \eqref{eq:Rmimoi}, we write
\begin{align}
  \MoveEqLeft I(\tilde{U}; Y, \hat{\eta}_1, \hat{\eta}_2 \cond S = \Sm,
  \hat{S} = \hat{\Sm}) \\
  &= I(\tilde{U}; \hat{\eta}_1 ) + I(\tilde{U}; Y, \hat{\eta}_2 \cond \hat{\eta}_1) \label{eq:tmp43} \\
  &= I(\tilde{U}; Y, \hat{\eta}_2 \cond \hat{\eta}_1) \label{eq:tmp44} \\
  &= I(\tilde{U}; Y - \hat{\eta}_1, \hat{\eta}_2 \cond \hat{\eta}_1) \label{eq:tmp45} \\
  &= I(\tilde{U}; \hv^\H \tilde{U} + {\Delta_1} + E, \hat{\eta}_2) \label{eq:tmp46} \\
  &= I(\tilde{U}; \hv^\H \tilde{U} + {\Delta_1} + E, a_{\text{mmse}}
  \, \gv^\H \tilde{U} + E_{\text{mmse}}) \label{eq:tmp47} \\
  &= \log\det( \Id + \Dm_1 \Sm \Qu
  \Sm^\H  )\label{eq:tmp48} 
\end{align}%
where \eqref{eq:tmp43} is from the chain rule of mutual information;
\eqref{eq:tmp44} is from the fact that $\tilde{U}$ is independent of
$\eta_1$; \eqref{eq:tmp46} holds because $\hat{\eta}_1$ is
independent of all the other terms. Since $\eta_2 = \hat{\eta}_2 + {\Delta_2}$ with  $\hat{\eta}_2
\sim \CN\bigl( 0, \gv^\H \Qu \gv (1 - \tilde{D}_2) \bigr)$ and
${\Delta_2} \sim \CN\bigl( 0,\gv^\H
\Qu \gv\,\tilde{D}_2 \bigr)$ being additive Gaussian noise, we can
optimally ``estimate'' $\hat{\eta}_2$ from $\eta_2$ with a linear MMSE
estimator and get the ``backward channel'' model
\begin{equation}
  \hat{\eta}_2 = a_{\text{mmse}}\, \eta_2 + e_{\text{mmse}}
\end{equation}%
where $a_{\text{mmse}} \defeq 1-\tilde{D}_2$ corresponds to the scaling of the
linear MMSE estimation and the additive estimation noise
$e_{\text{mmse}}\sim \CN\bigl(0, a_{\text{mmse}}\, \gv^\H \Qu
\gv\,\tilde{D}_2 \bigr)$ is
independent of the ``input'' $\eta_2$ of the estimator. Thus, \eqref{eq:tmp48} follows as the mutual
information of an equivalent Gaussian MIMO channel with Gaussian input,
where $\Qu = \Tu \Lu \Tu^\H$ and $\Qv = \Tv \Lv \Tv^\H$. Note that in
the right hand sides of the above equalities, we have omitted the
conditioning on $\bigl\{S = \Sm, \hat{S} = \hat{\Sm}\bigr\}$ for convenience of
presentation.  Finally, \eqref{eq:Rmimoi} follows from \eqref{eq:Rmimo1} and
\eqref{eq:tmp48}. Due to the symmetry, \eqref{eq:Rmimoj} is
straightforward. 
\end{proof}

Note that the optimization in \eqref{eq:R_scheme2} is not trivial and is
out of the scope of this paper. Instead of finding the exact rate, we
focus on the symmetric degrees of freedom of the scheme with $m=2$, by
fixing the following parameters:
\begin{align}
  \Qu &= \frac{P_1}{2} \Psimgp + \frac{P_2}{2} \Psimg, 
  \quad \Qv = \frac{P_1}{2} \Psimhp + \frac{P_2}{2} \Psimh, \\
  \Qm_\text{c} &= \frac{P_\text{c}}{2} \Id,\quad 
  \Qm_{\text{p}1} = \frac{P_\text{p}}{2} \Psimgp, \quad 
  \Qm_{\text{p}2} = \frac{P_\text{p}}{2} \Psimhp, \\
  \tilde{D}_1 &= \tilde{D}_2 = (P \sigma^2)^{-1} = P^{-(1-\alpha)}
  \label{eq:choice-distortion}
\end{align}%
where we recall that $\Psimg \defeq \displaystyle \frac{\hat{\gv}\hat{\gv}^\H}{\Norm{\hat{\gv}}}$ and
$\Psimgp$, $\Psimh$, and $\Psimhp$ are similarly defined;
the power allocations $(P_\text{c}, P_\text{p})$ and $(P_1, P_2)$ are specified by
\begin{flalign}
  &&P_\text{p} &= \hat{\alpha}\, \hat{\sigma}^{-2},& 
  P_\text{c} &= P - P_\text{p},& \\
  &&P_2 &= (1-\hat{\alpha}) \frac{P}{2} \hat{\sigma}^{2},& 
  P_1 &= P - P_2,&  
\end{flalign}%
with 
$\hat{\sigma}^2 \defeq \max\left\{ P^{-1}, \sigma^2 \right\}$
and
$\hat{\alpha} \defeq -\frac{\log\hat{\sigma}^2}{\log P}$. 
The interpretation of the choices on the covariance matrices has already
been given in Section~\ref{sec:proposed-scheme}. For the choices of the
distortions \eqref{eq:choice-distortion} and the power
allocations, the intuitions are as follows: 
\begin{itemize}
  \item The distortions $\tilde{D}_1$ and $\tilde{D}_2$ are such that the
    errors $\{\Delta_{k,t}\}$ after the reconstruction of $\eta_1$ and $\eta_2$ are at the noise level.  
  \item The transmit power of the private signals 
    scales as $P_\text{p} \sim P^{\alpha}$, while the received power at the
    unintended receiver scales as $P^{0}$, i.e., the noise level.
    Thus, the private signal does not incur any DoF loss for the
    unintended receiver.   
  \item The scaling factor $\hat{\alpha}$ ensures that $P_\text{p} = P$ and
    $P_\text{c} = 0$ when the estimation error is small, i.e., $\sigma^2 \le
    P^{-1}$ while leading to $P_\text{p} = 0$ and $P_\text{c} = P$ when
    the estimation error is high, i.e., $\sigma^2 = 1$. Similarly, with  
    $(1-\hat{\alpha})$, $P_1 = P$ and
    $P_2 = 0$ when the estimation error is small, while $P_1 = P_2 =
    \frac{P}{2}$ when the estimation error is high. 
\end{itemize}
It is readily shown that, with these choices, we have the high SNR
approximation of the rates
\begin{align}
  R_\text{c} &= (1-\alpha) \log P + O(1), \\
  R_{\text{p}k} &= \alpha \log P + O(1), \quad k=1,2, \\
  \Retav &= 2(1-\alpha) \log P + O(1), \\
  \Rmimok &= (2-\alpha) \log P + O(1), \quad k=1,2,
\end{align}%
from which we derive the symmetric DoF
  $d_\text{sym} = \frac{2+\alpha}{3}$.

\begin{biographynophoto}{Sheng Yang} (M'07) received the B.E. degree
  in electrical engineering from Jiaotong University, Shanghai, China,
  in 2001, and both the engineer degree and the M.Sc. degree in
  electrical engineering from \'Ecole Nationale Sup\'erieure des
  T\'el\'ecommunications (ENST), Paris, France, in 2004, respectively.
  From 2004 to 2007, he worked as teaching and research assistant in the
  Communications \& Electronics department in ENST. During the same
  period, he completed his Ph.D., graduating in 2007 from Universit\'e de
  Pierre et Marie Curie (Paris VI).  From October 2007 to November 2008,
  he was with Motorola Research Center in
  Gif-sur-Yvette, France, as a
  senior staff research engineer. Since December 2008, he has joined the
  Telecommunications department at SUPELEC where he is currently an
  assistant professor. His research interests include cooperative
  diversity schemes, wireless networks information theory, and
  coding/decoding techniques for multi-antenna communication systems.
\end{biographynophoto}

\begin{biographynophoto}{Mari Kobayashi} (M'06) received 
the B.E. degree in electrical engineering from
Keio University, Yokohama, Japan, in 1999 and the M.S. degree in mobile radio
and the Ph.D. degree from \'Ecole Nationale Sup\'erieure des T\'el\'ecommunications, Paris, France, in 2000 and 2005, respectively.
From November 2005 to March 2007, she was a Postdoctoral Researcher at
the Centre Tecnol\`ogic de Telecomunicacions de Catalunya, Barcelona, Spain.
Since May 2007, she has been an Assistant Professor at Sup\'elec, Gif-sur-Yvette,
France. Her current research interests include multiple-input-multiple-output
(MIMO) communication systems, multiuser communication theory.
\end{biographynophoto}

\begin{biographynophoto}{David Gesbert} (IEEE Fellow)
is Professor and Head of the Mobile Communications Department, EURECOM,
France. He obtained the Ph.D degree from Ecole Nationale Sup\'erieure
des T\'el\'ecommunications, France, in 1997. From 1997 to 1999, he has
been with the Information Systems Laboratory, Stanford University. In
1999, he was a founding engineer of Iospan Wireless Inc, San Jose, CA., a startup company pioneering MIMO-OFDM (now Intel). Between 2001 and 2003 he has been with the Department of Informatics, University of Oslo as an adjunct professor. D. Gesbert has published about 170 papers and several patents all in the area of signal processing, communications, and wireless networks.

 D. Gesbert was a co-editor of several special issues on wireless
 networks and communications theory, for JSAC (2003, 2007, 2009),
 EURASIP Journal on Applied Signal Processing (2004, 2007), Wireless
 Communications Magazine (2006). He served on the IEEE Signal Processing
 for Communications Technical Committee, 2003-2008.  He's an associate
 editor for IEEE Transactions on Wireless Communications and the EURASIP
 Journal on Wireless Communications and Networking. He authored or
 co-authored papers winning the 2004 IEEE Best Tutorial Paper Award
 (Communications Society) for a 2003 JSAC paper on MIMO systems, 2005
 Best Paper (Young Author) Award for Signal Proc. Society journals, and
 the Best Paper Award for the 2004 ACM MSWiM workshop. He co-authored
 the book ``Space time wireless communications: From parameter
 estimation to MIMO systems'', Cambridge Press, 2006.  
\end{biographynophoto}

\begin{biographynophoto}{Xinping Yi} (S'12) received his B.S. degree
   from Huazhong University of Science and Technology and M.Sc. degree
   from University of Electronic Science and Technology, China, both in
   Electrical Engineering. Currently, He is pursuing the Ph.D.~degree at
   Mobile Communication Department, EURECOM, Sophia Antipolis, France.
   From 2009 to 2011, he was a Research Engineer in Huawei Technologies,
   Shenzhen, China. His current research interests include multiuser
   information theory and signal processing.
\end{biographynophoto}

\end{document}